%% file: main.tex
\documentclass[a4paper,UKenglish,cleveref, autoref, thm-restate]{lipics-v2021}

\usepackage{cite}
\usepackage[table]{xcolor}
\usepackage{tikz}

\title{The Complexity of the Shapley Value for Regular Path Queries} %TODO Please add

\author{Majd Khalil}{Technion, Haifa, Israel}{smajd11@cs.technion.ac.il}{}{}
\author{Benny Kimelfeld}{Technion, Haifa, Israel}{bennyk@cs.technion.ac.il}{0000-0002-7156-1572}{}
\authorrunning{M.~Khalil and B.~Kimelfeld} %TODO mandatory. First: Use abbreviated first/middle names. Second (only in severe cases): Use first author plus 'et al.'

\Copyright{Majd Khalil and Benny Kimelfeld} %TODO mandatory, please use full first names. LIPIcs license is "CC-BY";  http://creativecommons.org/licenses/by/3.0/

 %TODO mandatory; please add comma-separated list of keywords

\nolinenumbers %uncomment to disable line numbering

\EventEditors{John Q. Open and Joan R. Access}
\EventNoEds{2}
\EventLongTitle{42nd Conference on Very Important Topics (CVIT 2016)}
\EventShortTitle{CVIT 2016}
\EventAcronym{CVIT}
\EventYear{2016}
\EventDate{December 24--27, 2016}
\EventLocation{Little Whinging, United Kingdom}
\EventLogo{}
\SeriesVolume{42}
\ArticleNo{23}
\usepackage{mathtools}

\input{macros.tex}


\begin{document}

\maketitle

\begin{abstract}
    A path query extracts from a labeled graph vertex tuples based on the words that are formed from the paths that connect the vertices. We study the computational complexity of measuring the contribution of edges and vertices to an answer of a path query. We focus on conjunctive regular path queries. To measure this contribution, we adopt the traditional Shapley value from cooperative game theory. This value has been recently proposed and studied in the context of relational database queries, and has uses in a plethora of other domains, including importance measurement of feature values for machine-learning classifiers.
    
    We first study edge contribution and show that the exact Shapley value is almost always hard to compute. Specifically,
    it is \#P-hard to calculate the contribution of an edge whenever at least one (non-redundant) conjunct allows for a word of length three or more. In the case of regular path queries (i.e., no conjunction),
    the problem is tractable if the query 
    has only words of length at most two; hence, 
    this property fully characterizes the tractability of the problem.
    On the other hand, if we allow for an approximation error, then it is straightforward to obtain an
    efficient scheme (FPRAS) for an additive approximation. 
    Yet, a multiplicative approximation is harder to obtain. 
    We establish that in the case of conjunctive regular path queries, 
    a multiplicative approximation of the
    Shapley value of an edge can be computed in polynomial time if and only if all query atoms are finite languages (assuming non-redundancy and conventional complexity limitations). 
    We also study the analogous situation where we wish to determine the contribution of a vertex, rather than an edge, and establish results of a similar nature.
\end{abstract}

\input{intro}

\input{prelim}

\input{shapley_pq}

\input{exact}
\input{approx}

\input{vertices}

\input{conclusion}

%%
%% Bibliography
%%

%% Please use bibtex, 

\bibliography{main}

\appendix

\end{document}

%% file: macros.tex
\usepackage[table]{xcolor}
\usepackage{amsmath}
\usepackage{tikz}
\usepackage{mathtools}
\usepackage{graphicx}% http://ctan.org/pkg/graphicx
\usepackage{subcaption}
\usepackage{apxproof}

\newtheoremrep{theorem}{Theorem}[section]
\newtheoremrep{lemma}{Lemma}[section]
\newtheoremrep{observation}{Observation}[section]

\def\e#1{\emph{#1}}
\def\lbl{\mathrm{lbl}}
\def\Sigmastar{\mathord{\Sigma^*}}

\def\angs#1{\langle #1\rangle}

\def\shapley{\mathrm{Shapley}}

\def\q#1{q_{#1}}

\def\pqshapleyval#1{\mathrm{Shapley}\angs{#1}}

\def\rpqshapley#1{\mathsf{RPQShapley}\angs{#1}}
\def\crpqshapley#1{\mathsf{CRPQShapley}\angs{#1}}
\def\rpqshapleyv#1{\mathsf{RPQShapley}^{\mathsf{v}}\angs{#1}}
\def\crpqshapleyv#1{\mathsf{CRPQShapley}^{\mathsf{v}}\angs{#1}}
\def\vpqv{v_{pq}^{\mathsf{v}}}

\def\dbshapley#1{\mathsf{DBShapley}\angs{#1}}
\def\dbshapleyval#1{\mathrm{Shapley}\angs{#1}}

\newcommand{\eat}[1]{}

\def\endo{_{\mathsf{n}}}
\def\exo{_{\mathsf{x}}}

\def\fpsharpp{\mathord{\mathrm{FP}^{\mathrm{\#P}}}}
\def\sharpp{\mathrm{\#P}}
\def\np{\mathrm{NP}}
\def\nl{\mathrm{\textsc{NLogspace}}}
\def\exp{\mathrm{\textsc{Expspace}}}
\def\pspace{\mathrm{\textsc{Pspace}}}
\def\ptime{\mathrm{\textsc{Ptime}}}

\newcommand*{\MyDef}{\mathrm{def}}
\newcommand*{\eqdefU}{\ensuremath{\mathop{\overset{\MyDef}{=}}}}% Unscaled version

\newcommand*{\eqdef}{\mathrel{\overset{\MyDef}{\resizebox{\widthof{\eqdefU}}{\heightof{=}}{=}}}}

\def\set#1{\mathord{\{#1\}}}

\def\qrst{Q_{\mathsf{RST}}}

\def\M{\mathcal{M}}
\DeclarePairedDelimiter\abs{\lvert}{\rvert}%

\newtheoremrep{claim}[theorem]{Claim}

\def\tup#1{\vec{#1}}

\def\prob{\mathrm{Pr}}

%% file: intro.tex
\section{Introduction}
\label{chap:intro}

Graph databases arise in common applications where the underlying data is a network of entities, and especially when connectivity and path structures are of importance. Such usage spans many fields, including the Semantic Web~\cite{DBLP:conf/pods/ArenasP11}, social networks~\cite{DBLP:conf/icdt/Fan12}, biological networks~\cite{DBLP:journals/biodatamining/LysenkoRSMRA16, yoon2017use}, data provenance~\cite{DBLP:conf/edbt/AnandBL10}, fraud detection~\cite{sadowski2014fraud}, recommendation engines~\cite{DBLP:conf/IEEEwisa/YiLFS17}, and many more.  In their simplest form, graph databases are finite, directed, edge-labeled graphs. Vertices represent entities, and edges of different labels are relationships of different types thereof. The connected nature of graph databases raises the need for tools that allow inspecting and analyzing the structure and patterns present in the data; this is typically realized in the form of queries that enable users to specify the structure of paths of interest.

A canonical example of a graph query is the Regular Path Query (RPQ)~\cite{DBLP:conf/sigmod/CruzMW87, DBLP:conf/pods/CalvaneseGLV99, DBLP:conf/pods/ConsensM90, DBLP:conf/pods/Yannakakis90}.  RPQs allow the specification of paths using regular expressions over the edge labels. When evaluated on a graph, the answers are source-target pairs of vertices that are connected by a path that conforms to the query's regular expression. This allows users to inspect complex connections in graphs by enabling them to form queries that match arbitrarily long paths. An important generalization of RPQs are the  \e{Conjunctive Regular Path Queries} (CRPQs) that extend regular path queries by allowing a conjunction of atoms where each atom is an RPQ that should hold between two specified variables~\cite{DBLP:conf/pods/ConsensM90, DBLP:conf/sigmod/CruzMW87}.

Being simple and expressive, such queries are an integral part of popular graph query languages for graphs, such as GraphLog, Cypher, XPath, and SPARQL. Therefore, they motivate and give rise to much research effort, including the study of some natural computational problems and variations thereof~\cite{DBLP:conf/pods/FlorescuLS98, DBLP:journals/siamcomp/MendelzonW95, DBLP:conf/icdt/MartensT18, DBLP:journals/tods/MartensT19}: 
%\begin{itemize}
% \item The decision problem:
What is the complexity of deciding whether an RPQ $r$ matches a path from $s$ to $t$ in $G$ (what we refer to as \e{Boolean query evaluation})?
Can we efficiently count and enumerate these paths?
Is a given CRPQ contained in another given CRPQ?
Importantly, the combined complexity of Boolean query evaluation is in $\ptime$ for RPQs, while for CRPQs it is $\np$-complete \cite{barcelo2012expressive}.
% \benny{Citation needed here!!!} 
Data complexity, however, is $\nl$-complete for both~\cite{DBLP:conf/pods/Baeza13}.  The containment problem for RPQs is $\pspace$-complete, and for CRPQs, it is $\exp$-hard~\cite{DBLP:conf/pods/FlorescuLS98, DBLP:conf/kr/CalvaneseGLV00}.

% Add part on related work

In this paper, we focus on the problem of \e{quantifying the responsibility and contribution} of different components in the graph, namely edges and vertices, to an answer of the CRPQ (and RPQ in particular). 
This problem has been studied in the context of queries on relational databases, and our motivation here is the same as in the relational context: we wish to provide the database user with an explanation of \e{why} (or what in the database led to that) we got a specific answer; when many combinations of data items can lead to an answer, and the lineage is too large or complex, we wish to quantify the contribution of individual items in order to distinguish between the more and the less important to the answer~\cite{DBLP:journals/corr/abs-2112-08874}. 

In the relational model, several definitions and frameworks have been proposed for measuring the contribution of a tuple. For example, Meliou et al.~\cite{DBLP:journals/pvldb/MeliouGMS11} defined the responsibility of a fact $f$ as, roughly, the inverse of the minimal number of facts needed to be removed to make $f$ counterfactual (i.e., the query answer is determined by the existence of $f$); this measure is an adaptation of earlier notions of formal causality by Halpern and Pearl~\cite{DBLP:conf/uai/HalpernP01}.
\e{Causal effect} is another alternative measure proposed by Salimi et al.~\cite{DBLP:conf/tapp/SalimiBSB16}:
% \benny{there should always be a citation over al., i.e., et al. [...]} 
if the database is probabilistic and each fact has independently the probability $1/2$ of existence, how does the probability of the answer change if we assume the existence or absence of $f$?  Lastly, and most relevant to our work, recent work has studied the adoption of the \e{Shapley value}---a solution concept from game theory~\cite{DBLP:journals/sigmod/LivshitsBKS21,DBLP:conf/pods/ReshefKL20,DBLP:conf/icdt/LivshitsK21,DBLP:journals/corr/abs-2112-08874}.

%\paragraph{The Shapley value.} 
The Shapley value is a formula for wealth distribution in a cooperative game~\cite{shapley:book1952}.  In databases, the conceptual application is straightforward: facts are the players who play the game of answering the Boolean (or numerical) query; hence, the wealth function is the result of the query~\cite{DBLP:journals/sigmod/LivshitsBKS21}.
%This makes it a useful tool for assigning shares of a jointly acquired reward between participants, so it has been used in many real-word cases like, 
The Shapley value has a plethora of applications, including profit sharing between ISPs~\cite{DBLP:journals/ton/MaCLMR10}, influence measurement in social network analysis~\cite{DBLP:journals/tase/NarayanamN11}, determining the most important genes for specific body functions~\cite{moretti2007class}, and identifying key players in terrorist networks~\cite{DBLP:journals/snam/CampenHHL18}, to name a few.
Closer to databases is a recent application to model checking for measuring the influence of formula components~\cite{mascle2021responsibility}. 
As another example, in machine learning, the SHAP score~\cite{NIPS2017_7062} has been used for measuring the contribution of each feature to the prediction, and it is essentially the Shapley value with the features as players.
This value was also used for quantifying the responsibility that every fact has on the \e{inconsistency} of a knowledge base~\cite{DBLP:journals/ai/HunterK10,DBLP:conf/ijcai/YunVCB18} and a 
database~\cite{DBLP:conf/icdt/LivshitsK21}.  Yet, what limits the applicability of the Shapley value is the high computational complexity---the execution cost might grow exponentially with the number of players.  Hence, past research has been investigating islands of tractability and approximation guarantees.
%  Which leads  to the search for non-trivial representations for the problems that allow better solutions or alternatively, resorting to approximations that are efficient. 

\paragraph*{Contribution}
We study the complexity of computing the Shapley value of edges and vertices in the domain of graph databases, where the queries are CRPQs.  In the remainder of this section and throughout the paper, we focus on edges (and discuss vertices in Section~\ref{sec:vertices}).  As in previous work~\cite{DBLP:journals/pvldb/MeliouGMS11, DBLP:conf/icdt/LivshitsK21}, we view the graph as consisting of two types of edges: \e{endogenous} edges and \e{exogenous} edges. The endogenous edges are the ones that we consider for reasoning about the contribution, and they are 
the players of the game.
The exogenous edges serve as external knowledge: they are static items that we take for granted
(as outside, unconcerned factors, deemed not to be possible causes), and do not participate in the counterfactual game. The classification into endogenous/exogenous items is application-dependent, and may even be chosen by the user at query time.

%We investigate the complexity of computing the Shapley values of edges in a graph to the answer of a CRPQ, in a later section we address the changes that need to be made when we reason about the contribution of vertices instead of edges.
% in a later section we also address the setting where we reason about the contribution of vertices.
% We are interested in the data complexity of the computational problems we mention. The input to the problem of computing the Shapley value in our context usually consists of two components. The first component
An instance of our problem involves a query $q$ (e.g., an RPQ or a CRPQ), an input graph $G$, an answer tuple $\tup t$ of vertices of $G$, and an edge $e$ of which contribution to $t$ we seek to measure. As previously done in the context of contribution measures, we adopt the yardstick of \e{data complexity}~\cite{DBLP:conf/stoc/Vardi82} where we consider the query  $q$ as fixed. More precisely, each fixed query $q$ is associated with a distinct computational problem that takes as input $G$, $\tup t$ and $e$.

We first show that the exact computation of the Shapley value is almost always hard. Specifically, it is sufficient for the CRPQ to have a non-redundant atom (i.e., a conjunct associated with a regular language) with a word of length three or more for the computation to be \#P-hard ($\fpsharpp$-complete).
% Intuitively, non-redundant atom means that it is interesting and adds information to the CRPQ.
In addition, for RPQs (i.e.,  single-atom CRPQs), we complete this hardness condition to a full dichotomy by showing that
the  Shapley value can be computed in polynomial time if the language contains only words of length at most two.

%we identify that for the family of queries where this hardness rule does not apply; 
% We identify however a limited set; CRPQs with a single atom, or in other words RPQs, where if the regular language only contains words of length 2 then computation of the Shapley values can be done in polynomial time.

Next, we study the complexity of approximation. In our context, we adopt a standard yardstick of tractable approximation, namely FPRAS (Fully Polynomial-Time Approximation Scheme). An approximation of the Shapley value of an edge to a CRPQ can be computed via a straightforward Monte-Carlo (average-over-samples) estimation of the expectation that Shapley defines. This estimation guarantees an additive (or absolute) approximation. However, we are also interested in a multiplicative (or relative) approximation.

We establish a dichotomy that classifies CRPQs into a class where there is a multiplicative FPRAS and the complementing class where there cannot be any such FPRAS under conventional complexity assumptions. Specifically, if the CRPQ contains an atom (non-redundant atom) with an infinite regular language, then (any) multiplicative approximation is intractable since it is already $\np$-complete to determine whether the Shapley value is nonzero. In every other case (assuming no redundant atoms), an additive FPRAS can also be used to obtain a multiplicative FPRAS, due to the \e{gap property}, previously established in the relational model~\cite{DBLP:conf/icdt/LivshitsK21,DBLP:conf/pods/ReshefKL20}: if the Shapley value is nonzero, it is at least the reciprocal of a polynomial. Note that this is contrasting the situation with relational conjunctive queries, where there is always a multiplicative FPRAS~\cite{DBLP:journals/sigmod/LivshitsBKS21}.

For the case of vertices, we show that the situation is very similar to what we have for edges. 
It is generally hard to compute exact values; it is sufficient for the CRPQ to have a non-redundant atom that contains a word of length four or more for the computation to be hard, while for RPQs we identify that the tractable family of queries is also tractable for vertices (yet, for vertices we do not complete a full classification). For approximation, we establish the same dichotomy as for edges.

\paragraph*{Paper organization.}
The rest of the paper is organized as follows. We introduce some basic terminology in Section~\ref{sec:prelims}. In Section~\ref{sec:shapley-edges}, we formally define how the Shapley value is applied in our setting for edges in graph databases. In Section~\ref{sec:exact}, we study the complexity of computing exact Shapley values for CRPQs, and we study approximations in Section~\ref{sec:approx}. We present the complexity results for the case when measuring contribution of vertices instead of edges in Section~\ref{sec:vertices} and outline the changes that should be made to the proofs. We conclude and discuss directions for future work in Section~\ref{sec:conclusion}.
For lack of space, some of the proofs are omitted or sketched in the body of the paper, and given fully in the Appendix.

%% file: prelim.tex
\section{Preliminaries}
\label{sec:prelims}
We begin by setting some terminology and notation that we use throughout the paper.

\subsection{Graphs and Path Queries}
We use $\Sigma$ to denote a finite alphabet (i.e., a finite set of symbols) that is used for labeling edges of graphs. A \e{word} is a finite sequence of symbols from $\Sigma$. As usual, $\Sigmastar$ denotes the set of all words. A \e{language} $L$ is a (finite or infinite) subset of $\Sigmastar$. By a slight abuse of notation, we may identify a language $L$ with a representation of $L$ such as a regular expression or a finite-state automaton.

A \e{regular expression} is defined as follows: $\emptyset$, $\epsilon$, and $\sigma \in \Sigma$ are regular expressions denoting the empty language, empty word and symbol $\sigma$, respectively; in addition, if $r$ and $s$ are regular expressions, then $(r \mid s)$ and $(r \cdot s)$ and $(r^*)$ are also regular expressions, denoting union, concatenation and Kleene star, respectively.  We sometimes omit parentheses and dots when there is no risk of ambiguity (so we may write $rs$ instead of $(r \cdot s)$, for instance).  The language $L(r)$ that $r$ accepts (recognizes) is defined as usual. We allow the use of a special regular expression $\Sigmastar$ that accepts every word.  A \e{deterministic finite automaton} (DFA) $A$ is a tuple $(Q, \Sigma, \delta, q_0, F)$, where $Q$ is a finite set of states, $\Sigma$ is a finite alphabet, $\delta \colon Q \times \Sigma \to Q$ is the transition function, $q_0$ is the initial state, $F$ is the set of accepting states. By $\delta^*(w)$ we denote the state that the automaton reaches after reading $w$, starting from the initial state. The automaton \e{accepts} a word $w$ if $\delta^*(w) \in F$.  We again use $L(A)$ to denote the language that $A$ recognizes. (Recall that the classes of regular expressions and DFAs coincide in their expressive power.)

By a \e{graph}  we mean an edge-labeled directed graph $G = (V, E)$ where $V$ is the finite set of nodes and $E \subseteq V \times \Sigma \times V$ is the set of edges. We will consistently denote by $n$ and $m$ the number of nodes and edges, respectively; that is, $n=\abs{V}$ and $m=\abs{E}$.  A path $p$ from node $u$ to node $v$ in $G$ is a sequence $p = (v_0, a_1, v_1)(v_1, a_2, v_2) \ldots (v_{k-1}, a_k, v_k)$ of edges in $G$ such that $u = v_0$ and $v = v_k$.  By $|p|$ we denote the length $k$ of $p$, and by $\lbl(p)$ we denote the word $a_1\cdots a_k$.
If $G = (V, E)$ is a graph and $E'\subseteq E$  is a set of edges, then we denote by $G[E']$ the subgraph $G' = (V, E')$ of $G$. In other words, $G[E']$ is obtained from $G$ by removing every edge in
$E\setminus E'$.

\input{graphics/figure_graph}

A \e{path query} is a query of the form $q \coloneqq (x, L, y)$ 
where $L$ is a language.  
When evaluated on a graph $G$ it returns the set $q(G)$ of all pairs $(s, t)$ for variables $(x, y)$, such that $s$ and $t$ are nodes in graph $G$ and there exists a path $p$ from $s$ to $t$ with $\lbl(p)\in L$.  For convenience, we may view $q$ as a function that also takes as input a pair of nodes, such that $q(G,s,t)=1$ if $(G,s,t)$ is a ``yes'' instance and $q(G,s,t)=0$ otherwise.  We define similarly the special case of the queries, which we call \e{regular path queries} (RPQs), where $L$ is a regular language that is defined via a regular expression $r$ or an automaton $A$. We sometimes use the shorthand $L$ for the query $(x, L, y)$, or $r$ in the case of a regular expression.

% Let $L$ be a language. We denote by $\pathq L$ the decision problem that takes as input a graph $G$, a source node $s$ and a target node $t$, and determines whether $G$ contains a path $p$ from $s$ to $t$ such that $\lbl(p)\in L$.  For convenience, we view $\pathq L$ as a function such that $\pathq L(G,s,t)=1$ if $(G,s,t)$ is a ``yes'' instance and $\pathq L(G,s,t)=0$ otherwise.  We define similarly the special case of the problem (and function) where $L$ is a \e{regular language} that is defined via a regular expression $\alpha$ or an automaton $A$, and then the corresponding notation is $\rpq{\alpha}$ and $\rpq{A}$, respectively.

% Let $L$ be a language. We denote by $\pathq L$ the function that takes as input a graph $G$ and return all pairs $(s, t)$ of nodes in the graph such that there exists a path $p$ from $s$ to $t$ such that $\lbl(p)\in L$.
% For convenience, we may view $\pathq L$ as a function that also takes as input a pair of nodes, such that $\pathq L(G,s,t)=1$ if $(G,s,t)$ is a ``yes'' instance and $\pathq L(G,s,t)=0$ otherwise.  We define similarly the special case of the functions where $L$ is a \e{regular language} that is defined via a regular expression $r$ or an automaton $A$, and then the corresponding notation is $\rpq{r}$ and $\rpq{A}$, respectively.

\begin{example} \label{example:queries}
Figure \ref{fig:graph_example} depicts the graph $G$ of our running examples with labels in $\Sigma=\{a, b, c\}$. We show a few examples of path queries on the graph $G$.

\begin{itemize} 
    \item $q_1 = \Sigmastar$. This query tests whether there is a path from $s$ to $t$ in $G$. For example, we have that $q_1(G, v_1, v_2) = 1$, and $q_1(G, v_1, v_6) = 1$, as there are paths from $v_1$ to both $v_2$ and $v_6$. But $q_1(G, v_3, v_1) = 0$, as there is no path from $v_3$ to $v_1$.
    \item $q_2 = \{abc\}$. This query tests whether there is a path from $s$ to $t$ in $G$ that matches the word $abc$. For example, we have that $q_2(G, v_1, v_6) = 1$, as there is a path $v_1 \rightarrow v_3 \rightarrow v_5 \rightarrow v_6$ that matches $abc$. But we have that $q_2(G, v_3, v_5) = 0$, as the only path from $v_3$ to $v_5$ consists of a single edge labeled $b$.
    \item $q_3 = ab^*$. This query tests whether there is a path from $s$ to $t$ in $G$ that matches regular expression $ab^*$. For example, we have that $q_3(G, v_1, v_6) = 1$, as there are paths: $v_1 \rightarrow v_2 \rightarrow v_4 \rightarrow v_6$, or alternatively, $v_1 \rightarrow v_2 \rightarrow v_6$, that match $ab^*$. But we have that $q_2(G, v_3, v_5) = 0$, as the only path from $v_3$ to $v_5$ consists of a single edge with label $b$, which does not match the regular expression.
\end{itemize}
Note that $q_1$ and $q_3$ are RPQs with infinite languages, whereas $q_2$ is defined simply as a finite (singleton) language.
\qed
\end{example}

\subsubsection{Conjunctive Regular Path Queries}\label{sec:CRPQs}
A \e{conjunctive regular path query}, CRPQ for short, is a query with $k$ variables $x_1, \ldots, x_k$ that is a conjunction of atomic regular path queries between pairs of nodes that can be assigned to $x_1, \ldots, x_k$. A general CRPQ $q$ is in the form:
\begin{equation}
q[x_1, \ldots, x_k] = \bigwedge_{i=1}^m (y_i, r_i, z_i) \label{eq:crpq}    
\end{equation}
Where $y_i$ and $z_i$ are variables from $\{x_1, \ldots, x_k \}$ and $r_i$ is a regular expression. The RPQ $(y_i, r_i, z_i)$ is also referred to as the atom $i$ of $q$ and is denoted by $q_i$.
As before, when evaluated on a graph $G$, we denote by $q(G)$ all of the assignments $(u_1, \ldots, u_k)$ for $(x_1, \ldots, x_k)$, such that all atomic RPQs are true.
We might also denote the assignment $(u_1, \ldots, u_k)$ as a function $\mu:\set{x_1,\dots,x_k}\rightarrow\set{u_1,\dots,u_k}$ such that $\mu(x_i)=u_i$ for $i=1,\dots,k$.
We use a numeric notation similarly to RPQs, that is, for a particular assignment $u_1, \ldots, u_k$ we have $q[u_1, \ldots, u_k](G)=1$ if it is a ``yes'' instance and $q[u_1, \ldots, u_k](G)=0$ otherwise.

\begin{example} \label{example:crpq}
Let us look at the query $q[x_1, x_2, x_3] = (x_1, a^*, x_2) \wedge (x_2, b^*, x_3)$. When evaluated on a graph, this query returns triplets $(u_1, u_2, u_3)$ such that there is a path from $u_1$ to $u_2$ that matches $a^*$ and there is a path from $u_2$ to $u_3$ that matches $b^*$.
In our example graph (Figure~\ref{fig:graph_example}), we have that $q[v_1, v_2, v_6](G) = 1$, as there is a path $v_1 \rightarrow v_2$ that matches $a^*$, and a path $v_2 \rightarrow v_4 \rightarrow v_6$ that matches $b^*$. Yet, $q[v_1, v_3, v_6](G) = 0$, as the only path from $v_3$ to $v_6$ has label $bc$, which does not match $b^*$.
% \benny{Put \qed in the end of examples like I do here.}
\qed
\end{example}

We say that an atom is \emph{redundant} if removing it from $q$ does not change the answers of $q$ on all graphs.
Formally, we denote by $q^{\setminus j}$ the query after removing atom $j$.
\[ q^{\setminus j}[x_1, \ldots, x_k] = \bigwedge_{i=1;i\neq j}^m (y_i, r_i, z_i) \]
Then the atom $j$ is \e{redundant} if
$ q \equiv q^{\setminus j} $, that is,
$q(G) = q^{\setminus j}(G) $ for all graphs $G$.

\begin{example} \label{example:redundant}
Let us look at $q[x_1, x_2, x_3] = (x_1, a, x_2) \wedge (x_2, b, x_3) \wedge (x_1, a^*b^*, x_3)$. In this query, the third atom is redundant according to our definition, as removing it does not change the result set on any graph. Intuitively, if the first two queries return true then so does the third, and so, the third atom does not add any restriction to the conjunction.
\qed
\end{example}

We will later refer to the following obvious (and standard) observation.
\begin{observationrep}
\label{obs:redundant}
Let $q$ be a CRPQ. If the atom $i$ is non-redundant, then there exists a graph $G$ and assignment $\mu$ to $(x_1,\dots,x_k)$ such that
$\q{j}(G, \mu(y_j) ,\mu(z_j))=1$ for $j\neq i$ and $\q{i}(G, \mu(y_i) ,\mu(z_i))=0$.
\end{observationrep}
% \bennyc{Something is wrong with the English of the observation. Also, the next paragraph should be a proof of the observation (proof environment).}
% \majd{Better?}

\begin{proof}
By definition, if atom $i$ is non-redundant, it means that there exists a graph $G_i$ such that $q(G_i) \neq q^{\setminus i}(G_i)$. And since every answer tuple $t \in q(G_i)$ is also also an answer tuple for $q^{\setminus i}(G_i)$, this means that there exists an answer tuple $t'=(v_1, \ldots, v_k) \in q^{\setminus i}(G_i)$ such that $t' \notin q(G_i)$. In other words, for assignment $(v_1, \ldots, v_k)$, there are matching paths for each $\q{j}$ for $j \neq i$ $(\q{j}(G_i, s_j ,t_j)=1)$ and no matching path for $\q{i}$ $(\q{i}(G_i, s_i ,t_i)=0)$.
% all RPQ atoms of $q$ return true in graph $G_i$ but the $i$th atom, i.e, $\q{j}(G_i, s_j ,t_j)=1$ for every $j\neq i$ and $\q{i}(G_i, s_i, t_i)=0$. 
\end{proof}
% Given an input graph $G$, source node $s$, target node $t$ and edge $e$ for $\rpqshapley{r_i}$, we show how to construct an input instance $G^*$ for $\crpqshapley{q}$. Since atom $i$ is non-redundant, by definition there exists graph $G_i$ and assignment $(v_1, \ldots, v_k)$ such that all RPQ atoms return true but the $i$th atom, i.e, $\q{r_j}(G_i, s_j ,t_j)=1$ for every $j\neq i$ and $\q{r_i}(G_i, s_i, t_i)=0$. 

% \bennyc{I am missing some background on the basic complexity of CRPQs. Give some background from the following paper and reference it: \url{https://dl.acm.org/doi/10.1145/2463664.2465216} }
% \majd{I did not fully understand this, did you mean to add more on it in the introduction? what do you mean by basic complexity?}

% \bennyc{I meant that you should say here what is the data/combined complexity of evaluating RPQs and CRPQs over graphs, with a reference to some place that says it.}

In the sequel, we say that $q$ is \e{without redundancy} if every atom of $q$ is non-redundant. Note that every CRPQ $q$ can be made one without redundancy (while preserving equivalence) by repeatedly removing redundant atoms.

\subsection{The Shapley Value}
Let $A$ be a finite set of players. A cooperative game is a function $ v\colon P(A) \to\mathbb{R}$ such that $v(\emptyset) = 0$. The value $v(S)$ represents a value, such as wealth, jointly obtained by $S$ when the players of $S$ cooperate. The Shapley value for the player $a$ is defined to be:
\begin{equation}\label{eq:shapley}
  \shapley(A, v, a) = \frac{1}{\lvert A \rvert !} \sum_{\sigma \in \Pi_{A}} (v(\sigma_{a} \cup \{a\}) - v(\sigma_{a}))
  \end{equation}
Here, $\Pi_{A}$ is the set of all possible permutations over the players in $A$, and for each permutation $\sigma$ we denote by $\sigma_a$ the set of players that appear before $a$ in the permutation.
Alternatively, the Shapley value can be written as follows. 
\[\shapley(A, v, a) = \sum_{B \subseteq A \setminus \{a\}} \frac{\abs{B} ! (\abs{A}-\abs{B}-1)! }{\abs{A} !} (v(B \cup \{a\}) - v(B)) \]
Intuitively, the Shapley value of a player $a$ is the expected contribution of $a$ to the value $v(B)$ where $B$ is a set of players chosen by randomly (and uniformly) selecting players one by one without replacement. The Shapley value is known to be unique up to some rationality axioms that we omit here (c.f.~\cite{shapley:book1952}).

%% file: graphics/figure_graph.tex
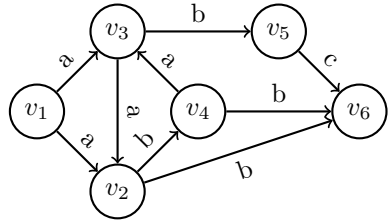
\begin{figure}[t]
\centering
\begin{tikzpicture}[node distance={15mm}, thick, main/.style = {draw, circle}] 
\node[main] (1) {$v_1$}; 
\node[main] (2) [below right of=1] {$v_2$}; 
\node[main] (3) [above right of=1] {$v_3$}; 
\node[main] (4) [above right of=2] {$v_4$}; 
\node[main] (5) [above right of=4] {$v_5$}; 
\node[main] (6) [below right of=5] {$v_6$}; 
\draw[->] (1) -- node[midway, above, sloped] {a} (2); 
\draw[->] (1) -- node[midway, above, sloped] {a} (3);
\draw[->] (3) -- node[midway, above, sloped] {a} (2);
\draw[->] (4) -- node[midway, above, sloped] {a} (3);
\draw[->] (2) -- node[midway, above, sloped] {b} (4);
\draw[->] (4) -- node[midway, above, sloped] {b} (6); 
\draw[->] (3) -- node[midway, above, sloped] {b} (5);
\draw[->] (2) -- node[midway, below, sloped] {b} (6); 
\draw[->] (5) -- node[midway, above, sloped] {c} (6); 
% \draw[->] (1) -- node[midway, above, sloped] {b} (3); 
% \draw (1) to [out=135,in=90,looseness=1.5] (5); 
% \draw (1) to [out=180,in=270,looseness=5] (1); 
% \draw (2) -- (4); 
% \draw (3) -- (4); 
% \draw (5) -- (4); 
% \draw[->] (5) to [out=315, in=315, looseness=2.5] (3); 
% \draw[->] (6) -- node[midway, above, sloped] {a} (4); 
\end{tikzpicture} 
\caption{The graph of the running example. In the text, we denote the edge $(v_i, v_j)$ as $e_{ij}$.}
\label{fig:graph_example}
\end{figure}

%% file: shapley_pq.tex
\section{The Shapley Value of Edges}
\label{sec:shapley-edges}

Throughout the paper, we focus on the Shapley value of edges of the input graph $G$. Later, in Section~\ref{sec:vertices}, we also discuss the extension of our results to the Shapley value of vertices.

Given a CRPQ $q$, our goal is to quantify the contribution of edges in the input graph $G$ to an answer $\tup u$ for $q$.  We adopt the convention that, for the sake of measuring contribution, the database is viewed as consisting of two types of data items---we reason about the contribution of the \e{endogenous} items while we take for granted the existence of the \e{exogenous} items (that serve as out-of-game background)~\cite{DBLP:conf/tapp/SalimiBSB16,DBLP:journals/sigmod/LivshitsBKS21,DBLP:journals/pvldb/MeliouGMS11}. Hence, in our setup, we view the graph as consisting of two types of edges: \e{endogenous edges} and \e{exogenous edges}.
% Exogenous edges $(E_x)$ are given and their existence in the graph is not questioned, while endogenous edges $(E_n)$ are what we are interested in their marginal contribution to the result.
Notationally, for a graph $G = (V, E)$ we denote by $E\endo$ and $E\exo$ the sets of endogenous and exogenous edges, respectively,
and we assume that $E$ is the disjoint union of $E\endo$ and $E\exo$.

Our goal is to quantify the contribution of an edge $e\in E\endo$ to an answer $\tup u = (u_1, \ldots u_k)$ of the query $q$, that is, to the fact that
$q[\tup u](G)=1$. To this end, we view the situation as a cooperative game where the players are the endogenous edges. The Shapley value of an edge $e \in E\endo$ in this setting will be denoted by $\pqshapleyval{q}(G,\tup u, e)$.
\[ \pqshapleyval{q}(G, \tup u, e) \eqdef \shapley(E\endo, v_{pq}, e) \] where the function $\shapley$ is as defined in Equation~\eqref{eq:shapley} and $v_{pq}$ is the numerical function that takes as input a subset of the endogenous edges and is defined as follows:
% $$
% v_{st}^{L}(E')=\begin{cases}
% 			1, & \text{there is a path from $s$ to $t$ in $G'=(V, E')$ that matches a word in $L$. }\\
%             0, & \text{otherwise.}
% 		 \end{cases}
% $$

\[
v_{pq}(B) \eqdef q[\tup u] (G[B \cup E\exo])-q[\tup u](G[E\exo])
\]
In particular, $v_{pq}(\emptyset)=0$. Put differently, we have the following.
\begin{multline*}
    \pqshapleyval{q}(G,\tup u, e) = \\ \sum_{B \subseteq E\endo \setminus \{e\}} \frac{\abs{B} ! (\abs{E\endo}-\abs{B}-1)! }{\abs{E\endo} !}  \Big(q[\vec u] (G[B \cup E\exo \cup \{e\}]) - q[\tup u](G[B \cup E\exo])\Big)
\end{multline*}

For a CRPQ $q$, the computational problem $\crpqshapley{q}$ is that of computing the Shapley value of a given edge:
% that of computing $\pqshapleyval{q}(G, \{u_1, \ldots, u_k\}, e)$.
\smallskip
\begin{center}
\begin{tabular}{|r|p{8cm}|}\hline
  \multicolumn{2}{|c|}{Problem $\crpqshapley{q}$}\\\hline
  Input: & Graph $G$, node vector $\tup u=(u_1, \ldots, u_k)$, endogenous edge $e$\\
  Goal: & Compute $\pqshapleyval{q}(G, \tup u, e)$ \\\hline
\end{tabular}
\end{center}
\smallskip

When $q$ has only one atom, and is in fact an RPQ $(x, r, y)$ with $r$ being a regular expression,
we may replace $q$ with $r$ in the notation and write $\pqshapleyval{r}(G, s, t, e)$ and 
$\rpqshapley{r}$ with the meaning of $\pqshapleyval{q}(G, s, t, e)$ and 
$\rpqshapley{q}$, respectively.

\begin{toappendix}
Here we show some examples for the computation of the Shapley values of edges for different inputs.

\begin{example}
We refer to the graph $G$ of Figure~\ref{fig:graph_example}. By default, all edges are endogenous unless stated otherwise.
(Note that we denote by $e_{ij}$ the edge $(v_i,v_j)$.)
\begin{itemize}
\item $\pqshapleyval{abc}(G, v_1, v_6, e)$. Any edge that is not on the only path that matches $abc$, namely $p \colon v_1 \rightarrow v_3 \rightarrow v_5 \rightarrow v_6$, will have the Shapley value of zero. This is true since adding such an edge to $G[B \cup E\exo]$, where $B \subseteq E\endo$, will not change the result of the  path query. For edges on the path $ p$, the computations are similar to each other and they all have the have same Shapley value. For one of them to change the query result, it needs to appear after both other edges in the permutation of $E\endo$. This happens in $\frac{9!}{3}$ of the overall $9!$ permutations. So we have:
%   \begin{align*}
%     &\pqshapleyval{abc}(G, v_1, v_6, e_{13}) = \pqshapleyval{abc}(G, v_1, v_6, e_{35}) = \\
%     & \pqshapleyval{abc}(G, v_1, v_6, e_{56}) = \frac{1}{3}\,.
%   \end{align*}
  \[
    \pqshapleyval{abc}(G, v_1, v_6, e_{13}) = \pqshapleyval{abc}(G, v_1, v_6, e_{35}) = \pqshapleyval{abc}(G, v_1, v_6, e_{56}) = \frac{1}{3}\,.
  \]
    
  If we assume that $e_{13}$ is exogenous, then the other two edges will \e{split} the contribution evenly. Then we get:
  \[\pqshapleyval{abc}(G, v_1, v_6, e_{35}) = \pqshapleyval{abc}(G, v_1,\allowbreak v_6, e_{56}) = \frac{1}{2}\,.\]

\item $\pqshapleyval{ab^*}(G, v_1, v_6, e)$. This is a more complicated computation, since there are two paths that match the regular expression $ab^*$ (as we have seen in Example~\ref{example:queries}).
  Again, any edge that is not on any of these paths will have a Shapley value of zero. But now, the contributions of the remaining edges is not equal since, for instance, $e_{12}$ is on both paths so we expect it to have higher contribution than the others. For the edge $e_{26}$ to change the query result, it needs to appear after edge $e_{12}$ but before at least one of $e_{24}$ and $e_{46}$. Permutations where this happens are either permutations where $e_{26}$ appears after $e_{12}$ and one of $e_{24}$ and $e_{46}$ but before the other one, and there are $2 \cdot \sum_{i=0}^{5} (i+2)!\binom{5}{i}(8-i-2)!=\frac{2}{12}\cdot 9!$ such permutations. Or permutations where $e_{26}$ appears after $e_{12}$ but before both $e_{24}$ and $e_{46}$, and there are $\sum_{i=0}^{5} (i+1)!\binom{5}{i}(8-i-1)!=\frac{1}{12} \cdot 9!$ such permutations. 
  There are an overall of $9!$ possible permutations, so,
  \[\pqshapleyval{ab^*}(G, v_1, v_6, e_{26}) = \frac{1}{4}\,.\]
  For the two edges $e_{24}$ and $e_{46}$, the computations are similar to each other.
  For one of them to change the query, it needs to appear after $e_{12}$ and the other one, but before $e_{26}$. Similar to before, there are $\sum_{i=0}^{5} (i+2)!\binom{5}{i}(8-i-2)!=\frac{1}{12}\cdot 9!$ permutations where this happens, so,
%   This happens in $9! \cdot \frac{1}{6} \cdot \frac{1}{2}$, and hence,
  \[\pqshapleyval{ab^*}(G, v_1, v_6, e_{24}) = \pqshapleyval{ab^*}(G, v_1, v_6, e_{46}) = \frac{1}{12}\,.\]
  For the last edge $e_{12}$, it needs to appear after $e_{26}$ or after both $e_{24}$, $e_{46}$. Permutations where this happens are either permutations where $e_{12}$ appears after both $e_{24}$ and $e_{46}$ but before $e_{26}$, after both $e_{26}$ and $e_{24}$ but before $e_{46}$, after both $e_{26}$ and $e_{46}$ but before $e_{24}$, and there are $3 \cdot \sum_{i=0}^{5} (i+2)!\binom{5}{i}(8-i-2)!=\frac{3}{12} \cdot 9!$ such permutations. Or permutations where $e_{12}$ appears after all three of $e_{26}, e_{24}, e_{46}$, and there are $\sum_{i=0}^{5} (i+3)!\binom{5}{i}(8-i-3)!=\frac{3}{12} \cdot 9!$ such permutations. Or permutations where $e_{12}$ appears after $e_{26}$ but before both $e_{24}$ and $e_{46}$, and there are $\sum_{i=0}^{5} (i+1)!\binom{5}{i}(8-i-1)!=\frac{1}{12} \cdot 9!$ such permutations. So overall we get that,
  \[\pqshapleyval{ab^*}(G, v_1, v_6, e_{12}) = \frac{7}{12}\,.\]
  Note that that $\sum_{e \in E\endo}{\pqshapleyval{ab^*}(G, v_1, v_6, e)} = 1$, which is expected since
  the sum of the Shapley values of all players is always equal to the value of the game for the whole set of players~\cite{shapley:book1952}.\qed
\end{itemize}
\end{example}
\end{toappendix}

%% file: exact.tex
\section{Complexity of Exact Computation}
\label{sec:exact}

% \bennyc{Do not start with a section title. Write some words in the beginning. Also, "Introduction" is overly general for this stage. Use something like "main result"}
In this section, we study the complexity of $\crpqshapley{q}$, where the goal is to compute the exact Shapley value of an edge.  Note that the query $q$ is fixed in the analysis, and hence, every $q$ defines a separate computational problem $\crpqshapley{q}$.

\subsection{Results}
Our main results for this section are the following theorems, showing that the $\crpqshapley{q}$ is computationally intractable for almost \e{every} CRPQ $q$, except for limited cases. 

% \begin{theorem}[Hardness]\label{thm:exact_hard}
% Let $q$ be a CRPQ without redundancies. If $q$ has a non-redundant atom $i$ with a language that  contains a word of length three or more, then $\crpqshapley{q}$ is $\fpsharpp$-complete.
% \end{theorem}

\begin{theorem}[Hardness]\label{thm:exact_hard}
Let $q$ be a CRPQ. If $q$ has a non-redundant atom $i$ with a language that  contains a word of length three or more, then $\crpqshapley{q}$ is $\fpsharpp$-complete.
\end{theorem}

Recall that $\fpsharpp$ is the class of functions computable in polynomial time with an oracle to a problem in $\sharpp$ (e.g., counting the number of satisfying assignments of a propositional formula).  This class is considered intractable, and above the polynomial hierarchy (Toda's theorem~\cite{DBLP:journals/siamcomp/Toda91}).

The question of whether the condition of Theorem~\ref{thm:exact_hard} is necessary for hardness remains open. Yet, we can show that it is, indeed, necessary, in the case of a single atom (RPQ):

% \sloppy
\begin{theoremrep}[Tractability]\label{thm:exact_tract}
Let $q$ be an RPQ with the regular expression $r$. If every word in $L(r)$ is of length at most two, then $\rpqshapley{q}$ is solvable in polynomial time.
\end{theoremrep}

\begin{proof}[Proof of Theorem~\ref{thm:exact_tract}]
As mentioned before, 
% in the case of one atom, the problem $\crpqshapley{q}$ is the same as $\rpqshapley{r}$, where $r$ is the only regular expression. So from equation .., 
the problem of computing Shapley values in our setting reduces to computing $\abs{\M(G, s, t, L, k)}$.
We will now show that 
% prove this by showing that 
$\abs{\M(G, s, t, L, k)}$ can be computed in polynomial time in this case. First, let us observe that we can compute  $\abs{\M(G, s, t, L, k)}$ by computing the complement set  $\abs{\overline{\M(G, s, t, L, k)}}$ which is defined similarly but for subsets where there is no path, and subtracting from the overall number of subsets of size $k$ of endogenous edges:
\[ \abs{\M(G, s, t, L, k)} = \binom{m_n}{k} - \abs{\overline{\M(G, s, t, L, k)}} \]
So it suffices to show how to compute $\abs{\overline{\M(G, s, t, L, k)}}$.

Let $L = \{ w_0, \ldots , w_l\}$. For a subset of endogenous edges to be in $\overline{\M(G, s, t, L, k)}$, it should not connect, with $E\exo$, any path from $s$ to $t$ matching $w \in L$, i.e., matching one of $w_0, \ldots, w_l$. In other words, it should not connect any path of size 2 matching some $\abs{w_i}=2$, or any path of size 1 matching some $\abs{w_i}=1$.

This divides endogenous edges of the graph into 4 categories:
\begin{itemize}
    \item $\mathsf{Permitted}$. Edges which are not part of any such path.
    \item $\mathsf{OnPath2E}$. Edges which are on a path of size 2 matching some $w_i$, where both edges are endogenous.
    \item $\mathsf{OnPath2X}$. Edges which are on a path of size 2 matching some $w_i$, where one of the edges is exogenous.
    \item $\mathsf{OnPath1}$. Edges which are on a path of size 1 matching some $w_i$.
\end{itemize}

These are disjoint paths, as we assume that there are no parallel edges, and edges on a path of size 1 can not also be on paths of size 2 and vice versa.

Each subset of size $k$ in $\overline{\M(G, s, t, L, k)}$ has no edges from $\mathsf{OnPath1}$ as such edges make a path, and no edges from  $\mathsf{OnPath2X}$ as such edges connect a path along with $E\exo$, $i$ edges from $\mathsf{Permitted}$, and $k-i$ edges from $\mathsf{OnPath2}$ such that no whole path is added, meaning a maximum of one edge is added from each path of size 2. For each $i$, this translates to $\binom{\abs{\mathsf{Permitted}}}{i}$ options for edges from $\mathsf{Permitted}$, and $2{\binom{\frac{\abs{\mathsf{OnPath2}}}{2}}{k-i}}$ options for edges from $\mathsf{OnPath2}$; choosing $k-i$ paths from $\frac{\abs{\mathsf{OnPath2}}}{2}$ paths, and for each path there are 2 possibilities for edges.
All in all we get that:
\[ \abs{\overline{\M(G, s, t, L, k)}} = \sum_{i=0}^{k} 2{\binom{\abs{\mathsf{Permitted}}}{i}} {\binom{\frac{\abs{\mathsf{OnPath2}}}{2}}{k-i}} \]
Which can be computed in polynomial time.
\end{proof}

Hence, we get a full classification for RPQs:
\begin{corollary}\label{cor:exact-rpq-dichotomy}
Let $q$ be an RPQ with the regular expression $r$. Assuming $\mathrm{P}\neq\mathrm{NP}$, the following are equivalent:
\begin{enumerate}
      \item $\crpqshapley{q}$ is solvable in polynomial time.
    \item Every word in $L(r)$ is of length at most two.
\end{enumerate}
% \benny{continue}
\end{corollary}

Next, in Section~\ref{sec:exact_hard} we prove Theorem~\ref{thm:exact_hard}, and in Section~\ref{sec:exact_tract} we prove Theorem~\ref{thm:exact_tract}.

\subsection{Proof of Hardness}
\label{sec:exact_hard}

Membership in $\fpsharpp$ is straightforward from the definition of the Shapley value in Equation~\eqref{eq:shapley}. Indeed,  $\pqshapleyval{q}(G, \tup u, e)$ can be computed using an oracle to the problem of counting the number of permutations over the edge set such that $e$ changes the query from zero (false) to one (true). For the $\fpsharpp$-hardness, we prove it in a sequence of reductions. We begin with hardness for the special case where the language consists of a single three-letter word. For that, we will use a result by Livshits et al.~\cite{DBLP:journals/sigmod/LivshitsBKS21} on the computation of Shapley values for facts in relational databases. We use that to prove hardness for the general case of a regular language (or any language) with one or more words of length at least three, even when restricted to simple graphs.

%\subsubsection{Databases and Shapley Value of Facts}
We first recall the result of Livshits et al.~\cite{DBLP:journals/sigmod/LivshitsBKS21}. They considered relational databases $D$ where some of the facts are endogenous and the rest exogenous. As in our notation, the corresponding subsets of $D$ are denoted by
$D\endo$ and $D\exo$, respectively.
For a Boolean query $q$ that maps every database into $\set{0,1}$, they defined the Shapley value of a fact similarly to the way we define the Shapley value of an edge: the endogenous facts are the players and the query is the wealth function:
$$\dbshapleyval{q}(D, f) = \shapley(D\endo, v_{db}, f)$$ where $v_{db}(E)=q(E \cup D\exo) - q(D\exo)$.  They established a complete classification of the class of conjunctive queries without self-joins into tractable and intractable queries for the computation of the Shapley value. What is relevant to us is that the following conjunctive query is $\fpsharpp$-hard:
\[\qrst()\colon\, \exists{x,y} [R(x)\land S(x,y)\land T(y)]\]

In addition, we define a special kind of graphs that will help us in some of the proofs.
A graph $G=(V, E)$ is called a \emph{leveled graph} if there exists a split of vertices into levels $V_0, \ldots, V_k$, such that:
\begin{enumerate}
    \item The set of vertices $V$ is the disjoint union of $V_0, \ldots, V_k$.
    \item Edges are only between vertices in consecutive levels.
\end{enumerate}

\eat{
\begin{center}
\begin{tabular}{|r|l|}\hline
  \multicolumn{2}{|c|}{Problem $\dbshapley{S, q}$}\\\hline
  Input: & Database $D$ over $S$, endogenous fact $f$\\
  Goal: & Compute the Shapley value of $f$ for $q(D)$\\\hline
\end{tabular}
\end{center}
}

%\begin{definition}[leveled graph] 

%\end{definition}

\input{graphics/figure_abc_reduction}

From the hardness of the Shapley value for 
$\qrst$ is is easy to prove the following. 
\begin{lemmarep}\label{lemma:sigma123}
Let $\sigma_i\in\Sigma$ for $i=1,2,3$.  $\rpqshapley{\sigma_1 \sigma_2 \sigma_3}$ is  $\fpsharpp$-hard, even when restricted to leveled graphs.
\end{lemmarep}
\begin{proof}
The proof is by reduction from the problem of computing $\dbshapleyval{\qrst}(D, f)$,
where $D$ is a database over scheme $\mathsf{S}=\{R(x), S(x, y), T(x)\}$.
Since, as mentioned earlier, this problem is $\fpsharpp$-hard.
Given a database $D$ over $\mathsf{S}$, we will construct a leveled graph database $G=(V, E)$ over $\Sigma=\{\sigma_1, \sigma_2, \sigma_3\}$.
We will map facts to edges using the following one-to-one mapping function (bijection), the matching edge for fact $f$ will be denoted by $e_f$: 
% \bennyc{It is not necessarily a bijection. What if $\sigma_1=\sigma_2$? Are you assuming that this is not the case? You need to be more careful here.}
For each fact $R(x)$ we will add an edge 
% from $s_0$ to $x_1$ with label $\sigma_1$, $e_f = 
$(s_0, \sigma_1, x_1)$.
For each fact $S(x, y)$ in we will add an edge 
% from $x_1$ to $y_2$ with label $\sigma_2$, $e_f = 
$(x_1, \sigma_2, y_2)$.
For each fact $T(y)$, we will add an edge 
% from $y_2$ to $t_3$ with label $\sigma_3$, $e_f = 
$(y_2, \sigma_3, t_3)$.
It can be seen that this graph is a three-level graph. It is constructed such that if the query is satisfied in database $D$, then there is a path from $s_0$ to $t_3$ matching $\sigma_1 \sigma_2 \sigma_3$. An edge $e_f$ is exogenous/endogenous according to the classification of the fact $f$.
To complete the proof, we will show that:
\[ \dbshapleyval{\qrst}(D, f) = \pqshapleyval{\sigma_1 \sigma_2 \sigma_3}(G, s_0, t_3, e_f) \]

% \bennyc{What are $s_0$ and $s_3$? You should say that upfront. Also, please add an example with a figure, the database and the graph.}

We will prove this by showing that both games are isomorphic. As stated before, the mapping function is a bijection and there is a one-to-one mapping between facts and edges, now we will show that the wealth functions get the same values for each subset $F$ of facts and its mapping to edges $E$ in the reduction graph $G$, i.e., $v_{db}(F) = v_{pq}(E)$. We denote by $\q{\sigma_1 \sigma_2 \sigma_3}$ the RPQ with regular expression $\sigma_1 \sigma_2 \sigma_3$.

Since $v_{db}(F) = \qrst(F \cup D\exo) - \qrst(D\exo)$ and $v_{pq}(E) = \q{\sigma_1 \sigma_2 \sigma_3} (G[E \cup E\exo], s_0, t_3)-\q{\sigma_1 \sigma_2 \sigma_3}(G[E\exo], s_0, t_3)$, it is sufficient to show that $\qrst(F') = \q{\sigma_1 \sigma_2 \sigma_3}(G[E'], s_0, t_3)$ for a subset $F'$ of facts and its mapping to edges $E'$. 
\begin{align*} 
\qrst(F') = 1 & \underset{(1)}{\iff} F' \models R(x), S(x, y), T(y) \\
& \underset{(2)}{\iff} \exists a, b : R(a), S(a, b), T(b) \in F' \\
& \underset{(3)}{\iff} (s_0, \sigma_1, a_1), (a_1, \sigma_2, b_2), (b_2, \sigma_3, t_3) \in E' \\
& \underset{(4)}{\iff} \q{\sigma_1 \sigma_2 \sigma_3}(G[E'], s_0, t_3) = 1 
\end{align*}

\sloppy
Transitions (1) and (2) follow from definitions, transition (3) follows from the construction of the reduction graph $G$, and transition (4) follows from that if $(s_0, \sigma_1, x_1), (x_1, \sigma_2, y_2), (y_2, \sigma_3, t_3) \in E$ then these edges connect a path from $s_0$ to $t_3$ matching $\sigma_1 \sigma_2 \sigma_3$.

So as a conclusion, for any fact $f$ we can compute its Shapley value efficiently, in polynomial time, given a solution for the second problem. The construction of the graph can be done in polynomial time.
\end{proof}

\eat{
\begin{proofsketch}
  The proof is by reduction from the problem of computing $\dbshapleyval{\qrst}(D, f)$ where $D$ is a database over scheme $\mathsf{S}=\{R(x), S(x, y), T(x)\}$.  Since, as mentioned earlier, this problem is $\fpsharpp$-hard.  Given a database $D$ over $\mathsf{S}$, we construct a leveled graph $G=(V, E)$ with labels from $\{\sigma_1, \sigma_2, \sigma_3\}$.  
  We map facts to edges as follows.  For each fact $R(a)$ we add an edge $(s, \sigma_1, a)$, for each fact $S(a,b)$ we add the edge $(a, \sigma_2, b)$, and for each fact $T(b)$ we add the edge $(b, \sigma_3, t)$.  We denote by $e_f$ the edge that we add for the fact $f$.  We make $e_f$ endogenous if and only if $f$ is endogenous (and exogenous if and only if $f$ is exogenous).  
  Clearly, $e_f$ is a bijection between the
  facts of $D$ and the edges of $G$, and $G$ is a three-level graph.
  Note that for every subset $D'$ of $D$ it holds that $D'$ satisfies $\qrst$ if and only if
  the subgraph $G[E']$ includes a path from $s$ to $t$ that matches $\sigma_1 \sigma_2 \sigma_3$, where
  $E'=\set{e_f\mid f\in D'}$.  With these we show that $\dbshapleyval{\qrst}(D, f) = \pqshapleyval{\sigma_1 \sigma_2 \sigma_3}(G, s, t, e_f) $, which completes the proof. 
  \end{proofsketch}
}
The proof (given in the Appendix) is via the reduction illustrated in Figure~\ref{fig:abc_reduction}.
% \benny{Should be a lemma, since you already have the more general theorem.}
Next, we show the following generalization of Lemma~\ref{lemma:sigma123}:
\begin{lemmarep} 
\label{lemma:exact_hardness}
% Given a problem $\pathq{L}$, if there exists $w \in L$: $\abs{w} \geq 3$, then $\pqshapley{L}$ is $\#P$-hard.
Let $r$ be a regular expression.
If there exists a word in $L(r)$ of length at least three, then $\rpqshapley{r}$ is $\fpsharpp$-hard, even when restricted to leveled graphs.
\end{lemmarep}

\begin{proof}
Let $w_0=\sigma_1 \ldots \sigma_n \in L(r)$, such that $n=\abs{w_0} \geq 3$. We will show a reduction from the problem $\rpqshapley{\sigma_1 \sigma_2 \sigma_3}$ on leveled graphs to $\rpqshapley{r}$ on leveled graphs.
% \benny{Use English. We will show a reduction from... to...  Given the input... for ..., we construct an input... for...}
Given an input graph $G$,
% \benny{But why? You haven't shown anything for leveled graphs already? All you have at this moment is the hardness for $\sigma_1\sigma_2\sigma_3$. Restructure. Show a reduction directly from the 3 case. Also, refer to the reader to the figure for an illustration.}
source node $s$, target node $t$ and edge $e$ for $\rpqshapley{\sigma_1 \sigma_2 \sigma_3}$, we construct an input graph $G'$ for $\rpqshapley{r}$. It will be constructed in the following way: 

Starting from $s$, we will keep edges in the first level with label $\sigma_1$, edges in the second level with label $\sigma_2$, and edges in the third level with label $\sigma_3$. 
These edges will be classified as exogenous/endogenous according to the classification of the original edges. All other edges will be removed. If $\abs{w} > 3$, we will also add a path consisting of exogenous edges, from node $t$ to a new node $t'$ matching the remainder of $w$, i.e., $\sigma_4 \ldots \sigma_n$. Such a reduction is illustrated in Figure~\ref{fig:leveled_reduction_example}. Now we will show that $\pqshapleyval{\sigma_1 \sigma_2 \sigma_3}(G, s, t, e) = \pqshapleyval{r}(G', s, t', e)$.

Now we can see that in the new graph $G'$ there can only be paths matching $w_0 = \sigma_1 \ldots \sigma_n$ from $s$ to $t'$; starting with a sub-path that was originally in $G$ that matches $\sigma_1 \sigma_2 \sigma_3$, and getting to $t'$ through the only path from $t$ to $t'$ that we added. In addition, any edge that was removed $e'$ has $\pqshapleyval{\sigma_1 \sigma_2 \sigma_3}(G, s, t, e') = 0$, since it does not participate in any path from $s$ to $t$ matching $\sigma_1 \sigma_2 \sigma_3$, otherwise it would have been added to $G'$ by construction. So adding such an edge to $G[B \cup E\exo]$, where $B \subseteq E\endo$, will not change the result of the  path query.
% \benny{Why so? Explain.}

Since the sub-path that matches the suffix $\sigma_4 ... \sigma_{n}$ of $w_0$, consists only of exogenous edges and can only be from $t$ to $t'$, the existence of a path from $s$ to $t'$ matching $r$ boils down to the existence of a path from $s$ to $t$ matching $\sigma_1 \sigma_2 \sigma_3$.
% the path query reduces to a path from $s$ to $t$ matching $\sigma_1 \sigma_2 \sigma_3$ instead of a path from $s$ to $t'$ matching $r$. 
Pairing that with the fact the in both graphs we have the same set of endogenous edges, except for endogenous edges in the original graph that were not added to $G'$, those edges can not be in a path from $s$ to $t$ matching $\sigma_1 \sigma_2 \sigma_3$ by construction, so they have a Shapley value of 0 and can be removed from the game without affecting values of the others, we get that $\pqshapleyval{\sigma_1 \sigma_2 \sigma_3}(G, s, t, e) = \pqshapleyval{r}(G', s, t', e)$.
\end{proof}

\input{graphics/figure_leveled}

With Lemma~\ref{lemma:exact_hardness}, we can prove Theorem~\ref{thm:exact_hard}.

\begin{proof}[Proof of Theorem~\ref{thm:exact_hard}]
  We know that there exists an atom $i$ of $q$ such that is non-redundant and $L(r_i)$ contains a word of length at least three. We reduce $\rpqshapley{r_i}$ on leveled graphs (Lemma~\ref{lemma:exact_hardness}) to $\crpqshapley{q}$.  Given an input graph $G$, source node $s$, target node $t$ and edge $e$ for $\rpqshapley{r_i}$, we construct an input instance $G^*$ for $\crpqshapley{q}$. Since the atom $i$ is non-redundant, we can use Observation~\ref{obs:redundant} and conclude that there exists a graph $G_i$ and assignment $\tup v$ to $\tup x$ 
%   \benny{to ...}
  such that all RPQ atoms return true except for the $i$th atom; that is, we have that $\q{j}(G_i, s_j ,t_j)=1$ for every $j\neq i$ and $\q{i}(G_i, s_i, t_i)=0$. Here, $s_j$ and $t_j$ are the nodes assigned to the variables $y_j$ and $z_j$, respectively, from Equation~\eqref{eq:crpq}.
% \bennyc{This is not clear. Where do you define non-redundancy?}

% To construct $G^*$, we will connect the graph $G$ to $G_i$.
We will assume that $G$ and $G_i$ are disjoint, and combine them to construct $G^*$ by merging
% \benny{We will assume that $G$ and $G_i$ are disjoint, and combine them by merging...} This will be made by merging 
nodes $s$ and $t$ in $G$ with nodes $s_i$ and $t_i$ in $G_i$, respectively (the ingoing edges to one are the ingoing edges to merged node, and the same goes for outgoing edges). Edges from $G_i$ will be exogenous and edges from $G$ will be classified as exogenous or endogenous according to the original edge.

% with all edges exogenous, and we will connect to it the graph $G$ by connecting source node $s$ and target node $t$ to nodes $s_i$ and $t_i$, respectively.
% Connecting as in both are the same node in the graph. 

% To construct $G^*$, we will have $G_i$ with all edges exogenous, and we will connect to it the graph $G$ by connecting source node $s$ and target node $t$ to nodes $s_i$ and $t_i$, respectively.
% Connecting as in both are the same node in the graph. Edges in $G$ will be classified as exogenous/endogenous according to original edge.
% \bennyc{the last three sentences make no sense to me. What is $G_i$? What is the grammar of the second sentence? What is the classification?}

The satisfaction of $q[\tup v]$ in $G^*$ 
is determined by the satisfaction of the atom $i$ for $s_i$ and $t_i$, since $G_i$ has a matching path for the atom $j$ for all $j\neq i$.
% \benny{Majd - what is ``atom RPQ $j$''? and ``of CRPQ $q$''? You mean ``the atom $j$''? ``the RPQ atom $j$''? ``of \textbf{the} CRPQ $q$''?
%   Please fix all over the paper.}
% $q[v_1, \ldots, v_k](G^*)$ 
%already boils down to the satisfaction of the atom $i$. \benny{Not clear what ``boils down'' means here. Entails?}
%\majd{Not sure if entails is better in this case, I'm trying to say that the satisfaction of $q$ reduces to the satisfaction of atom $i$, what do you think?}
% $\q{r_i}(G^*, s_i, t_i)$. \benny{Please rewrite the last sentence. Do not use math,math.}  
Recall that there are no paths from $s_i$ to $t_i$ matching $r_i$ in $G_i$. Hence, from our construction of $G^*$ (and in particular given that $s_i$ and $t_i$ are not part of any cycle), we get that every path from $s_i$ to $t_i$ that matches $r_i$ should be fully contained in $G$.
% the such paths can only be in $G$.\benny{What are the paths you refer to in ``such paths''?} 
%That means that the satisfaction of the RPQ atom $i$ of $q$ boils down to the satisfaction of $\q{r_i}(G, s, t)$; whether there is a path from $s$ (or $s_i$) to $t$ (or $t_i$) in $G$ that matches $r_i$. \bennyc{What is this notation?}
We conclude that $q[\tup v]$ is true in $G^*$ if and only if the RPQ atom $i$ is true in $G$, and the same holds if we remove any set of endogenous edges from both $G$ and $G^*$.
%
% Since also, $s \equiv s_i$ \benny{What does this equivalence mean?} and $t \equiv t_i$ \benny{and this one, what does it mean?}, we get that the query reduces to $\q{r_i}(G, s, t)$. 
%
%Since we also have the same set of endogenous edges in $G$ and $G^*$, we get that
Therefore, $\pqshapleyval{r_i}(G, s, t, e) = \pqshapleyval{q}(G^*, \tup v, e)$, as claimed.
% \bennyc{The query? Aren't we reducing one Shapley to another?}
%\bennyc{Therefore, the satisfaction of $q$ by a subset of the endogenous edges boils down to the satisfaction of $r_i$ by the same set of edges}
From Lemma~\ref{lemma:exact_hardness}, we know that $\rpqshapley{r_i}$ is $\fpsharpp$-hard , and hence, $\crpqshapley{q}$ is $\fpsharpp$-hard. 
\end{proof}
This completes the proof of the hardness side of Theorem~\ref{thm:exact_hard}. 
Next, we show the tractability side.

\subsection{Proof of Tractability}\label{sec:exact_tract}
We now discuss the idea of our polynomial-time algorithm for computing $\rpqshapley{r}$ where $L=L(r)$ consists of words of length at most two. 
We denote by $\M(G, s, t, L, k)$ the set of all subsets $E'$ of 
$E\endo$ such that $G[E\exo\cup E']$ contains a path of $L$ from $s$ to $t$.
If we group subsets of edges of the same size, we can also get the following form for the Shapley value:
% \begin{align*}
% \rpqshapley{r}(G, s, t, e) & = \sum_{k=0}^{\abs{E_n}-1} \frac{k ! (\abs{E_n}-k-1)! }{\abs{E_n} !} \abs{\M(G, s, t, L(r), k)}\\ &
%  - \sum_{k=0}^{\abs{E_n}-1} \frac{k ! (\abs{E_n}-k-1)! }{\abs{E_n} !} \abs{\M(G \setminus \{e\}, s, t, L(r), k)}
% \end{align*}
\begin{align*}
\rpqshapley{r}(G, s, t, e)  =  \sum_{k=0}^{\abs{E\endo}-1} & {\binom{\abs{E\endo}}{k}}^{-1} \abs{\M(G_e, s, t, L, k)}\\ &
 - \sum_{k=0}^{\abs{E\endo}-1} {\binom{\abs{E\endo}}{k}}^{-1} \abs{\M(G \setminus \{e\}, s, t, L, k)}
\end{align*}
Where $G_e$ is the same as $G$, except for $e$ that is exogenous instead of endogenous.

% \bennyc{I am missing something... How do you make sure that $e$ belongs to the sets of the first sum?}
% \majd{I'm not sure I understand this, for the first term we compute on the full graph, for the second one we compute on the graph without edge $e$, maybe i'm missing something.}

% \begin{align*}
% \rpqshapley{r}(G, s, t, e) & = \sum_{k=0}^{\abs{E_n}-1} {\binom{\abs{E_n}}{k}}^{-1} \abs{\M(G, s, t, L(r), k)}\\ &
%  - \sum_{k=0}^{\abs{E_n}-1} {\binom{\abs{E_n}}{k}}^{-1} \abs{\M(G \setminus \{e\}, s, t, L(r), k)}
% \end{align*}
\sloppy
This shows that the computation of $\rpqshapley{r}(G, s, t, e)$ reduces efficiently to computing $\abs{\M(G, s, t, L, k)}$, that is, counting the subsets of $E\endo$ 
(of endogenous edges) of size $k$ that, when added to $E\exo$, connects $s$ to $t$ via a path that matches a word in $w \in L$.
% \bennyc{You are not consistent with $L$ and $L(r)$. Choose one and use it all along.}
% \majd{It is just that the algorithm is for a general language $L$, and we use it for }
% We will make use this form later in the analysis. 
An algorithm that computes $\abs{\M(G, s, t, L, k)}$ efficiently in our case is provided in the Appendix. Using the realization above and the algorithm for computing $\abs{\M(G, s, t, L, k)}$ we can prove Theorem~\ref{thm:exact_tract}.

%% file: graphics/figure_abc_reduction.tex
\usetikzlibrary{positioning, calc}

\tikzstyle{leveled_node}=[draw, circle, minimum size=15pt,inner sep=1pt, scale=0.9]
% \tikzset{
% ->, % makes the edges directed
% node distance=2cm, % specifies the minimum distance between two nodes. Change if necessary.
% }
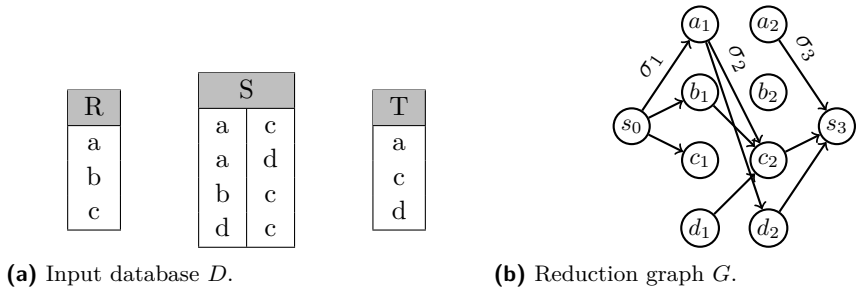
\begin{figure}[t]
\begin{subfigure}[b]{0.45\textwidth}
\centering
% \begin{minipage}[b]{.33\textwidth}\centering
\begin{minipage}[b]{0.3\linewidth}\centering
\begin{tabular}{ |c|  }
\hline
\rowcolor{lightgray}\multicolumn{1}{|c|}{R} \\
\hline
a \\
b \\
c \\
\hline
\end{tabular}
\end{minipage}
\begin{minipage}[b]{0.3\linewidth}\centering
\begin{tabular}{ |c|c|  }
\hline
\rowcolor{lightgray}\multicolumn{2}{|c|}{S} \\
\hline
a & c \\
a & d \\
b & c \\
d & c \\
\hline
\end{tabular}
\end{minipage}
\begin{minipage}[b]{0.3\linewidth}\centering
\begin{tabular}{ |c|  }
\hline
\rowcolor{lightgray}\multicolumn{1}{|c|}{T} \\
\hline
a \\
c \\
d \\
\hline
\end{tabular}
\end{minipage}
\caption{Input database $D$.}
\end{subfigure}
\begin{subfigure}[b]{0.45\textwidth}
\centering
\begin{tikzpicture}[node distance={10mm}, thick, main/.style = {draw, circle}] 

\node[leveled_node] (a1) {$a_1$};

\node[leveled_node] (b1) [below of=a1]{$b_1$};
\node[leveled_node] (c1) [below of=b1]{$c_1$};
\node[leveled_node] (d1) [below of=c1]{$d_1$};

\coordinate (Middle1) at ($(b1)!0.5!(c1)$);

\node[leveled_node] (s0) [left of=Middle1]{$s_0$};

\node[leveled_node] (a2) [right of=a1]{$a_2$};
\node[leveled_node] (b2) [below of=a2]{$b_2$};
\node[leveled_node] (c2) [below of=b2]{$c_2$};
\node[leveled_node] (d2) [below of=c2]{$d_2$};
% \node[leveled_node] (e2) [below of=d2]{$e_2$};

\coordinate (Middle2) at ($(b2)!0.5!(c2)$);

\node[leveled_node] (s3) [right of=Middle2]{$s_3$};

% \node[leveled_node] (6) [right of=3]{$v_6$};
% \node[leveled_node] (7) [below of=6]{$v_7$};
% \node[leveled_node] (8) [below of=7]{$v_8$};

% \node[leveled_node] (9) [right of=6]{$v_9$};

\draw[->] (s0) -- node[midway, above, sloped] {$\sigma_1$} (a1); 
\draw[->] (s0) -- (b1);
\draw[->] (s0) -- (c1);

\draw[->] (a1) -- node[near start, above, sloped] {$\sigma_2$} (c2);
\draw[->] (a1) -- (d2);
\draw[->] (b1) -- (c2);
\draw[->] (d1) -- (c2);

\draw[->] (c2) -- (s3);
\draw[->] (d2) -- (s3);
\draw[->] (a2) -- node[near start, above, sloped] {$\sigma_3$} (s3);

% (2) edge[above, midway, sloped] node{a} (5)
% (2) edge[above, midway, sloped] node{b} (4)
% % (2) edge[above, midway, sloped] node{c} (3)

% (3) edge[above, midway, sloped] node{d} (6)
% (3) edge[above, midway, sloped] node{b} (7)
% (3) edge[above, midway, sloped] node{b} (8)
% (4) edge[above, midway, sloped] node{c} (8)
% (5) edge[above, midway, sloped] node{b} (8)

% (6) edge[above, midway, sloped] node{c} (9)
% (7) edge[above, midway, sloped] node{e} (9)
% (8) edge[above, midway, sloped] node{c} (9)
% ;

\end{tikzpicture} 
\caption{Reduction graph $G$.}
\end{subfigure}

\caption{An example for the construction in the reduction of the proof of Lemma \ref{lemma:sigma123}.}
\label{fig:abc_reduction}
\end{figure}

%% file: graphics/figure_leveled.tex
\tikzstyle{leveled_node}=[draw, circle, minimum size=15pt,inner sep=2pt, scale=0.9]

\tikzset{
->, % makes the edges directed
node distance=2cm, % specifies the minimum distance between two nodes. Change if necessary.
}

\begin{figure}[b]
% \centering

\begin{subfigure}[b]{0.45\textwidth}
\centering
\begin{tikzpicture}[node distance={12mm}, thick, main/.style = {draw, circle}] 
\node[leveled_node] (1) {$v_1$};
\node[leveled_node] (2) [below of=1]{$v_2$};
% \node[leveled_node] (3) [below of=2]{$v_3$};

\node[leveled_node] (3) [right of=1]{$v_3$};
\node[leveled_node] (4) [below of=3]{$v_4$};
\node[leveled_node] (5) [below of=4]{$v_5$};
% \node[leveled_node] (7) [below of=6]{$v_7$};

\node[leveled_node] (6) [right of=3]{$v_6$};
\node[leveled_node] (7) [below of=6]{$v_7$};
\node[leveled_node] (8) [below of=7]{$v_8$};

\node[leveled_node] (9) [right of=6]{$v_9$};

\draw 
(1) edge[above] node{a} (3)
(1) edge[above, midway, sloped] node{a} (4)

(2) edge[above, midway, sloped] node{a} (5)
(2) edge[above, midway, sloped] node{b} (4)
% (2) edge[above, midway, sloped] node{c} (3)

(3) edge[above, midway, sloped] node{d} (6)
(3) edge[above, midway, sloped] node{b} (7)
(3) edge[above, midway, sloped] node{b} (8)
(4) edge[above, midway, sloped] node{c} (8)
(5) edge[above, midway, sloped] node{b} (8)

(6) edge[above, midway, sloped] node{c} (9)
(7) edge[above, midway, sloped] node{e} (9)
(8) edge[above, midway, sloped] node{c} (9)
;

\end{tikzpicture} 
\caption{Input graph $G$.}
% \vspace{1em}
\end{subfigure}
\begin{subfigure}[b]{0.45\textwidth}
\centering
\begin{tikzpicture}[node distance={12mm}, thick, main/.style = {draw, circle}] 
\node[leveled_node] (1) {$v_1$};
\node[leveled_node] (2) [below of=1]{$v_2$};
% \node[leveled_node] (3) [below of=2]{$v_3$};

\node[leveled_node] (3) [right of=1]{$v_3$};
\node[leveled_node] (4) [below of=3]{$v_4$};
\node[leveled_node] (5) [below of=4]{$v_5$};
% \node[leveled_node] (7) [below of=6]{$v_7$};

\node[leveled_node] (6) [right of=3]{$v_6$};
\node[leveled_node] (7) [below of=6]{$v_7$};
\node[leveled_node] (8) [below of=7]{$v_8$};

\node[leveled_node] (9) [right of=6]{$v_9$};

\node[leveled_node] (10) [right of=9]{};

\node[leveled_node] (11) [right of=10]{$t'$};

\draw 
(1) edge[above] node{a} (3)
(1) edge[above, midway, sloped] node{a} (4)

(2) edge[above, midway, sloped] node{a} (5)
% (2) edge[above, midway, sloped] node{b} (4)
% (2) edge[above, midway, sloped] node{c} (3)

% (3) edge[above, midway, sloped] node{d} (6)
(3) edge[above, midway, sloped] node{b} (7)
(3) edge[above, midway, sloped] node{b} (8)
% (4) edge[above, midway, sloped] node{c} (8)
(5) edge[above, midway, sloped] node{b} (8)

(6) edge[above, midway, sloped] node{c} (9)
% (7) edge[above, midway, sloped] node{e} (9)
(8) edge[above, midway, sloped] node{c} (9)

(9) edge[above, midway, sloped] node{d} (10)
(10) edge[above, midway, sloped] node{e} (11)

;
\end{tikzpicture} 
% \caption{Graph $G'$ constructed in the reduction}
\caption{Reduction graph $G'$}
\end{subfigure}

\caption{An example for the reduction in Lemma~\ref{lemma:exact_hardness}, for a regular expression that accepts the word $abcde$, source node $s=v_1$, target node $t=v_9$.}
\label{fig:leveled_reduction_example}
\end{figure}
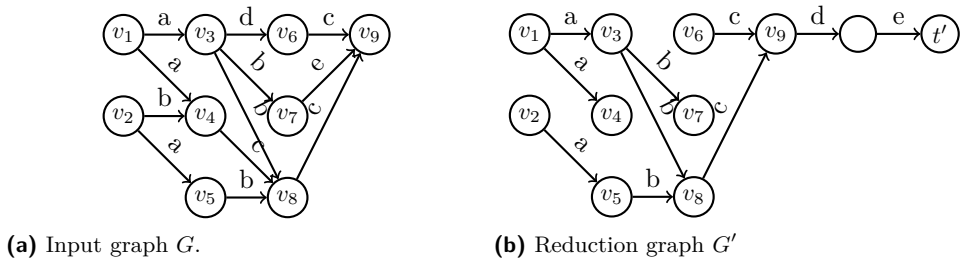

%% file: approx.tex
\section{Complexity of Approximation}
\label{sec:approx}

We now study the complexity of approximating $\crpqshapley{q}$. We aim for a \e{fully polynomial randomized approximation scheme}, or FPRAS for short. Formally, an FPRAS for a numeric function $f$ is a randomized algorithm $A(x, \epsilon, \delta)$, where $x$ is an input for $f$ and $\epsilon, \delta \in (0,1)$, that returns an $\epsilon$-approximation of
$f(x)$ with probability $1 - \delta$ (where the probability is over the randomness of $A$) in time polynomial in $x$, $1/\epsilon$ and $\log(1/\delta)$. We distinguish between an \emph{additive} FPRAS:
\[ \prob\left[ f(x) - \epsilon \leq A(x, \epsilon, \delta) \leq f(x) + \epsilon \right] \geq 1 - \delta \]
and a \emph{multiplicative} FPRAS:
\[ \prob\left[ \frac{f(x)}{1+\epsilon} \leq A(x, \epsilon, \delta) \leq (1+\epsilon) f(x) \right] \geq 1 - \delta \]

\subsection{Results}
Our main result for this section is a simple Monte-Carlo based algorithm that guarantees an additive approximation for any CRPQ, that also serves as a multiplicative FPRAS in some cases that we present in a dichotomy for when a given CRPQ admits a multiplicative approximation.
We note that here and later on, we sometimes give results for general CRPQs, yet without redundancy. These results generalize to CRPQs with redundant atoms by application to any CRPQ obtained by repeatedly eliminating redundancy (as mentioned in Section~\ref{sec:CRPQs}).

% An approximation of the Shapley value of an edge to a CRPQ can be computed via a straightforward Monte-Carlo (average-over-samples) estimation of the expectation that Shapley defines. This estimation guarantees an additive (or absolute) approximation. However, we are also interested in a multiplicative (or relative) approximation.

% Our main result for this section is the following dichotomy for when a given CRPQ admits a multiplicative CRPQ, also showing in the way a simple algorithm that serves as an additive FPRAS for any CRPQ.
% \benny{Mention clearly the existence of an additive FPRAS, Monte Carlo, like in the Introduction. Here you just hint to it indirectly.}

% \begin{theorem}\label{thm:approx}
%   Let $q$ be a CRPQ. If $L(r_i)$ is finite for every non-redundant atom  $i$ of $q$ , then $\crpqshapley{q}$ has a multiplicative FPRAS.
%   Otherwise, $\crpqshapley{q}$ has no multiplicative approximation (of any ratio) or else $\mbox{NP}\subseteq\mbox{BPP}$.
% \end{theorem}

\begin{theorem}\label{thm:approx}
  Let $q$ be a CRPQ without redundancy. If $L(r_i)$ is finite for every atom $i$ of $q$, then $\crpqshapley{q}$ has a multiplicative FPRAS.
  Otherwise, $\crpqshapley{q}$ has no multiplicative approximation (of any ratio) or else $\mbox{NP}\subseteq\mbox{BPP}$.
\end{theorem}
% \benny{What about the complexity assumption? Please check throughout the paper that you mention the complexity assumptions when needed.}

In the remainder of this section, we prove Theorem~\ref{thm:approx}, starting with the hardness side (Section~\ref{sec:approx_hard}) and  moving on to the FPRAS algorithm (Section~\ref{sec:approx_tract}).

\subsection{Proof of Hardness} \label{sec:approx_hard}

The hardness side is based on the following lemma that gives a characterization of when the Shapley value of an edge is nonzero.
\begin{lemma}\label{lemma:characterization-nonzero}
Let $G$ be a graph, $s$ and $t$ two vertices of $G$, and $e$ an endogenous edge of $G$. $\pqshapleyval{\Sigmastar}(G, s, t, e) > 0$ 
if and only if $e$ belongs to a simple path from $s$ to $t$.
\end{lemma}

\begin{proof}
We handle each direction separately.

$\Longleftarrow$: If $\pqshapleyval{\Sigmastar}(G, s, t, e) > 0$ then there exists some subset of edges $S$ that adding $e=(x, y)$ to it connects some path from $s$ to $t$, otherwise the marginal contribution of $e$ to all subsets of edges is zero and we get that $\pqshapleyval{\Sigmastar}(G, s, t, e) = 0$. We argue that adding $e$ connects at least one path from $s$ to $t$ that is simple. Since adding $e$ to the subset of edges $S$ connects a path from $s$ to $t$, then there already exist two sub-paths, $l_1$ from $s$ to $x$, and $l_2$ from $y$ to $t$ with all edges in $S$. If $l_1$ is not simple, we can get a simple path $l_1'$ by removing cycles from $l_1$, same applies for $l_2$. The path that combines $l_1', e, l_2'$ is a simple path from $s$ to $t$, since $l_1'$ and $l_2'$ are simple, in addition, let us assume that the path visits some vertex in $l_2'$ that it already visited in $l_1'$, then in contradiction to that $e$ has non-zero marginal contribution to $S$, we can get rid of the cycle that we have, and get a path with all edges in $S$, meaning that $e$ was not needed to connect such a path.

$\Longrightarrow$: $e$ lies on a simple path $l$ from $s$ to $t$ in $G$. If we look at $L$ the set of edges in path $l$ not including $e$ then adding $e$ to that subset of edges connects a path from $s$ to $t$, that path matches some word $w \in \Sigmastar$. So the marginal contribution of $e$ to that subset is $1$ and that means that $\pqshapleyval{\Sigmastar}(G, s, t, e) > 0$.
\end{proof}
A direct consequence of the characterization of Fortune, Hopcroft and Wyllie~\cite{DBLP:journals/tcs/FortuneHW80} of the subgraph homeomorphism problem is that the graph problem of Lemma~\ref{lemma:characterization-nonzero}
is $\np$-complete.

\begin{lemmarep}
\label{lemma:simplepath}
It is $\np$-complete to determine, given a graph $G$, nodes $s$ and $t$, and edge $e$, whether $e$ lies on any simple path from $s$ to $t$.
% $\edgesimplepath$ is $\np$-hard.
\end{lemmarep} 
\begin{proof}
In~\cite{DBLP:journals/tcs/FortuneHW80}, the set of pattern graphs for which the fixed directed sub-graph homeomorphism problem is $\np$-complete is characterized. An immediate result of that is that the decision problem that takes as input, a graph $G$, source node $s$, target node $t$, and node $v$, and decides whether $v$ lies on a simple path from $s$ to $t$ in graph $G$, is $\np$-hard. We will show a simple reduction from the problem of determining whether a node lies on a simple path to the problem of whether an edge lies on a simple path. Given a graph $G$, source node $s$, target node $t$, and node $v$, we modify $G$ such that node $v$ is split into two nodes $v_{in}$ and $v_{out}$ and we connect them by an edge $(v_{in}, v_{out})$. All ingoing edges into $v$ will be added as ingoing edges to $v_{in}$, and all outgoing edges from $v$ will be added as outgoing edges from $v_{out}$, the new graph will be denote by $G'$. We now argue that $v$ lies on a simple path from $s$ to $t$ in $G$ if and only if edge $(v_{in}, v_{out})$ lies on a simple path from $s$ to $t$ in $G'$, which completes the reduction.
\end{proof}

Hence, from Lemmas~\ref{lemma:characterization-nonzero} and~\ref{lemma:simplepath} we conclude that:
\begin{corollary}\label{cor:sigmastar-nonzero}
It is $\np$-complete to determine, given $G$, $s$, $t$ and $e$, whether $\pqshapleyval{\Sigmastar}(G, s, t, e) > 0$.
\end{corollary}

\input{graphics/figure_reduction}

Next, we generalize Corollary~\ref{cor:sigmastar-nonzero} from $\Sigmastar$ to any arbitrary infinite regular language $r$.

\begin{lemmarep}
\label{lemma:hardness_approx}
Let $r$ be a regular expression. If $L(r)$ is infinite, then it is $\np$-complete to determine whether $\pqshapleyval{r}(G, s, t, e) > 0$.
\end{lemmarep}
% \benny{If it is in NP, then you should cover full NP-completeness, not just NP-hardness. If there is any fundamental problem with proving it, then you should clearly explain it. Please fill the gap.}

\begin{proof}
We will prove this by showing a reduction from the problem of determining whether $\pqshapleyval{\Sigmastar}(G, s, t, e) > 0$ which we showed to be $\np$-hard.
Given an input instance $(G, s, t, a)$, we will show how to construct an instance $(G', s' , t', a_k)$ for our problem such that:
\[ \pqshapleyval{\Sigmastar}(G, s, t, a) > 0 \iff \pqshapleyval{r}(G', s', t', a_k) > 0\]
Since $L(r)$ is infinite, we know that its corresponding DFA graph that we will denote by $G_{DFA}$ has at least one cycle. We find a path from an initial state to an accepting state that passes through a node $v_i$ that is part of a cycle. We will denote the path by: $l: v_0 \rightarrow \ldots \rightarrow v_i \rightarrow \ldots \rightarrow v_k$.

We assumed the node $v_i$ is part of a cycle, we denote the labels which are along the cycle starting from $v_i$ by $w_{cycle} = \sigma_0 ... \sigma_c$.
The graph $G'$ will be constructed so that it any path in it from $s'$ to $t'$ matches $r$ in the following way, containing 3 sub-graphs:
\begin{itemize}
    \item The path $s'=v_0 \rightarrow \ldots \rightarrow v_i^{in}=s$, with the same labels as in the DFA, the edge $(v_{i-1}, v_i^{in})$ will have the label of $(v_{i-1}, v_i)$. (\e{exogenous}).
    \item A copy of the graph $G$ where each edge $e$ is split into $c$ edges with labels matching $w_{cycle}$ which we will denote by $e_1 ... e_c$. The node $s=v_i^{in}$ will serve as the source $s$ in the original graph, and the node $t=v_i^{out}$ will serve as the target $t$ in the original graph. (\e{endogenous/exogenous according to original edge}).
    \item The path $t=v_i^{out} \rightarrow \ldots \rightarrow v_k$, with the same labels as in the DFA, the edge $(v_i^{out}, v_{i+1})$ will have the label of $(v_i, v_{i+1})$. (\e{exogenous}).
\end{itemize}

We will now show that for each $a \in E$ and $a_k \in E'$ that is any edge that sits on the path that replaced the edge $a$ in $G'$ this holds: 
\[ \pqshapleyval{\Sigmastar}(G, s, t, a) > 0 \iff \pqshapleyval{r}(G', s', t', a_k) > 0\]

\underline{$\Longrightarrow:$}

% $ \pqshapleyval{\Sigmastar}(G, s, t, a) > 0 \Longrightarrow \pqshapleyval{r}(G', s', t', a_k) > 0$
% \noindent $\Longrightarrow:$
$\pqshapleyval{\Sigmastar}(G, s, t, a) > 0 \underset{(1)}{\Longrightarrow} \exists S \subset E\endo:\q{\Sigmastar} (G[S \cup E\exo \cup \{a\}],s,t) - \q{\Sigmastar}(G[S \cup E\exo],s,t) > 0$
$ \underset{(2)}{\Longrightarrow} $ Adding $a$ to $S \cup E\exo$ connects a path from $s$ to $t$ in $G$ 
% $\underset{\eqref{claim1_approx}}{\Longrightarrow}$
$\underset{(3)}{\Longrightarrow}$
Adding $a_k$ to $S' \cup E\exo'$ connects a path from $s$ to $t$ in $G'$ that matches $r$, where $S'$ is the subset containing all edges $e_i$ for each edge $e \in S$ and edges $a_i$ for $i \neq k$ $\underset{(4)}{\Longrightarrow} \exists S' \subset E\endo':\q{r}(G'[S' \cup E\exo' \cup \{a_k\}],s',t') - \q{r}(G'[S' \cup E\exo'],s',t') > 0$
$\underset{(5)}{\Longrightarrow} \pqshapleyval{r}(G', s', t', a_k) > 0$.\\

Transitions $(1), (2), (4), (5)$, follow from definitions. We need to prove transition $(3)$; 
% What we need to prove is the middle transitions, other transitions from definitions, we need to show that 
if adding $a$ to $S \cup E\exo$ connects a path from $s$ to $t$ in $G$, then, adding $a_k$ to $S' \cup E\exo'$ connects a path from $s'$ to $t'$ in $G'$ that matches $r$, where $S'$ is the subset containing all edges $e_i$ for each edge $e \in S$ and edges $a_i$ for $i \neq k$.

% \begin{claim}\label{claim1_approx}
% Adding $a$ to $S \cup E\exo$ connects a path from $s$ to $t$ in $G$ \bennyc{Use English instead of arrows. "Therefore, "}, therefore, adding $a_k$ to $S' \cup E\exo'$ connects a path from $s'$ to $t'$ in $G'$ that matches $r$, where $S'$ is the subset containing all edges $e_i$ for each edge $e \in S$ and edges $a_i$ for $i \neq k$.
% \end{claim}

% \begin{proof}
First, we will need to prove that there is no path from $s'$ to $t'$ in $G'$ using only $S' \cup E\exo'$ matching $r$. Let us assume that there is, then this path contains a sub-path from $s$ to $t$ of edges $e_i$, for each such edge $e_i$, $e \in S$ by the way of reduction, so we get that there is a path from $s$ to $t$ in $G$ using $S \cup E\exo$ by contradiction to that $a$ is needed to connect such path.
Now we will prove that adding $a_k$ connects a path from $s'$ to $t'$ matching $r$. We know that adding $a$ to $S \cup E\exo$ connects a path from $s$ to $t$ in $G$ so from that we get that adding $a_k$ to $S' \cup E\exo'$ connects a path from $s$ to $t$ in $G'$ that is the same path as in $G$ but each edge is split into $c$ edges, the overall path from $s'$ to $t'$ will be composed the only path from $s'$ to $s$ and then the path we mentioned from $s$ to $t$ and finally the path from $t$ to $t'$. We will now show that this path matches $r$ by simulating a run of the DFA for $r$ and showing that it ends in an accepting state.

Starting from the initial state $v_0$, the path from $s'$ to $s$ would transition the automaton into state $v_{i}$, as this path was constructed according to the transitions of the automaton. 
Then, each pass through a collection of edges $e_k$ that match an edge $e$ in the original graph keeps the automaton in the same state, this is because going through their labels creates a cycle. Since the path from $s$ to $t$ goes through a finite number of such transitions, we get that simulating the path until $t$ keeps the automaton in state $v_{i}$.
Finally, simulating the last part is similar to the first part, and it would transition the automaton into state $v_k$ which is an accepting state.
This means that adding $e_k$ connects a path from $s'$ to $t'$ matching $r$.

\underline{$\Longleftarrow:$}

% \noindent $\Longleftarrow:$
$ \pqshapleyval{r}(G', s', t', a_k) > 0
\underset{(1)}{\Longrightarrow} \exists S' \subset E\endo':\q{r}(G'[S' \cup E\exo' \cup \{a_k\}],s',t') - \q{r}(G'[S' \cup E\exo'],s',t') > 0 
\underset{(2)}{\Longrightarrow}$ Adding $a_k$ to $S' \cup E\exo'$ connects a path from $s'$ to $t'$ in $G'$ that matches $r$
$\underset{(3)}{\Longrightarrow}$ Adding $a$ to $S \cup E\exo$ connects a path from $s$ to $t$ in $G$, where $S$ is the subset containing all edges $e \in E$ s.t. $\forall i \in \{1, \ldots, c\}: e_i \in S'$
$\underset{(4)}{\Longrightarrow} \exists S \subset E\endo:\q{\Sigmastar}(G(S \cup E\exo \cup \{a\}),s,t) - \q{\Sigmastar}(G(S \cup E\exo),s,t) > 0 
\underset{(5)}{\Longrightarrow} \pqshapleyval{\Sigmastar}(G, s, t, a) > 0$.\\

Again, transitions $(1), (2), (4), (5)$, follow from definitions. We need to prove transition $(3)$; if adding $a_k$ to $S' \cup E\exo'$ connects a path from $s'$ to $t'$ in $G'$ that matches $r$, then, adding $a$ to $S \cup E\exo$ connects a path from $s$ to $t$ in $G$, where $S$ is the subset containing all edges $e \in E$ s.t. $\forall i \in \{1, ..., c\}: e_i \in S'$.

% \begin{claim}\label{claim2_approx}
% Adding $a_k$ to $S' \cup E\exo'$ connects a path from $s'$ to $t'$ in $G'$ that matches $r$, therefore, adding $a$ to $S \cup E\exo$ connects a path from $s$ to $t$ in $G$, where $S$ is the subset containing all edges $e \in E$ s.t. $\forall i \in \{1, ..., c\}: e_i \in S'$.
% \end{claim}

% \begin{claimproof}[Claim Proof]
First, we will need to prove that there is no path from $s$ to $t$ in $G$ using only $S \cup E\exo$. Let us assume that there is, then as shown before there is a path from $s'$ to $t'$ in $G'$ using only $S' \cup E\exo'$ that matches $r$, by contradiction to that $a_k$ is needed to connect such path.
Now we will prove that adding $a$ connects a path from $s$ to $t$. We know that adding $a_k$ to $S' \cup E\exo'$ connects a path from $s'$ to $t'$ in $G'$ that matches $r$, specifically, it connects a sub-path from $s$ to $t$. We can see that if such sub-path passes through some $e'_{k}$, it passes through all $e'_i$ so all $e'_i \in S'$, and consequently $e' \in S$. This means that adding $a$ to $S \cup E\exo$ connects a path from $s$ to $t$.
% \end{claimproof}
\end{proof}

\begin{proofsketch}
It is straight-forward to show that the problem is in $\np$ for a start, as any subset of endogenous edges that adding $e$ to it connects a matching path serves as a witness and can be verified in polynomial time.
We will prove $\np$-hardness by showing a reduction from the problem of determining whether $\pqshapleyval{\Sigmastar}(G, s, t, e) > 0$ which we showed to be $\np$-complete.
Given an input instance $(G, s, t, a)$, we will show how to construct an instance $(G', s' , t', a_k)$ for our problem such that:
\[ \pqshapleyval{\Sigmastar}(G, s, t, a) > 0 \iff \pqshapleyval{r}(G', s', t', a_k) > 0\]
Since $L(r)$ is infinite, we know that its corresponding DFA graph that we will denote by $G_{DFA}$ has at least one cycle. We find a path from an initial state to an accepting state that passes through a node $v_i$ that is part of a cycle. We will denote the path by: $l: v_0 \rightarrow \ldots \rightarrow v_i \rightarrow \ldots \rightarrow v_k$.

We assumed the node $v_i$ is part of a cycle, we denote the labels which are along the cycle starting from $v_i$ by $w_{cycle} = \sigma_0 \dots \sigma_c$.
The graph $G'$ will be constructed so that it any path in it from $s'$ to $t'$ matches $r$ in the following way, containing 3 sub-graphs:
\begin{itemize}
    \item The path $s'=v_0 \rightarrow \ldots \rightarrow v_i^{in}=s$, with the same labels as in the DFA, the edge $(v_{i-1}, v_i^{in})$ will have the label of $(v_{i-1}, v_i)$. (\e{exogenous}).
    \item A copy of the graph $G$ where each edge $e$ is split into $c$ edges with labels matching $w_{cycle}$ which we will denote by $e_1 ... e_c$. The node $s=v_i^{in}$ will serve as the source $s$ in the original graph, and the node $t=v_i^{out}$ will serve as the target $t$ in the original graph. (\e{endogenous/exogenous according to original edge}).
    \item The path $t=v_i^{out} \rightarrow \ldots \rightarrow v_k$, with the same labels as in the DFA, the edge $(v_i^{out}, v_{i+1})$ will have the label of $(v_i, v_{i+1})$. (\e{exogenous}).
\end{itemize}

We complete the proof by showing that that for each $a \in E$ and $a_k \in E'$ that is any edge that sits on the path that replaced the edge $a$ in $G'$ it holds that 
$\pqshapleyval{\Sigmastar}(G, s, t, a) > 0$ if and only if $\pqshapleyval{r}(G', s', t', a_k) > 0$.
\end{proofsketch}

% \benny{You cannot nest proofs inside each other! Split the proof into parts.}

Finally, we extend Lemma~\ref{lemma:hardness_approx} from RPQs to CRPQs.
\begin{lemmarep}\label{lemma:approx_hardness}
  Let $q$ be a CRPQ without redundancy.
   If $L(r_i)$ is infinite for some atom $i$ of $q$, then determining whether $\pqshapleyval{q}(G, \tup u, e) > 0$ is $\np$-complete.
 \end{lemmarep}
%  \benny{Again, NP-complete?}

The proof here is very similar to the proof we had for Theorem~\ref{thm:exact_hard}, and is provided in the Appendix for completeness.
%  \benny{Say something about the proof. It is annoyingly missing.}
 
\begin{proof}
It is straightforward to show that the problem is in $\np$.
For hardness, we know that there exists an atom $i$ of $q$ such that it is non-redundant and $L(r_i)$ is infinite. We will show a reduction from the problem of deciding whether $\pqshapleyval{r_i}(G, s, t, e) > 0$ which we showed to be $\np$-hard and that proves the lemma.

Given an input graph $G$, source node $s$, target node $t$ and edge $a$, the reduction would work the same way as in the proof of Theorem~\ref{thm:exact_hard} and we construct a $G^*$ that is a union of the graph $G_i$ that we get from that atom $i$ is non-redundant along with an assignment $\tup v$ (Observation~\ref{obs:redundant}), and the original graph $G$. 
The same claims hold as in the previous proof and we get:
\[ \pqshapleyval{r_i}(G, s, t, a) = \pqshapleyval{q}(G^*, \tup v, a) \]

That concludes the proof of the reduction as given a solution for the problem of determining whether $\pqshapleyval{q}(G, \tup u, e) > 0$ we can solve the problem of determining whether $\pqshapleyval{r_i}(G, s, t, e) > 0$ in polynomial time.
\end{proof}

\eat{
\begin{theorem} 
$\rpqshapley{r}$, where $L(r)$ is infinite, does not have an FPRAS.
\end{theorem}
}

\paragraph*{Open Problem: Directed Acyclic Graphs}
% \benny{Add here}
It is worth noting that the proof as shown in this section, does not work when the graph is acyclic as it relies on Lemma~\ref{lemma:simplepath} as a basis. Which states that the decision problem that takes as input, a graph $G$, source node $s$, target node $t$, and edge $e$, and decides whether $e$ lies on a simple path from $s$ to $t$ in graph $G$, is $\np$-complete. While this is true in the case of a general graph $G$, it is not when the graph is acyclic, where the problem can be solved in polynomial time. This leaves the problem of whether there is a different dichotomy when restricted to DAGs open. However, for exact computation there is no change even when restricted to DAGs as the proofs work as is.

\subsection{Proof of Tractability}
\label{sec:approx_tract}

We will now show that for any query that the condition for hardness does not hold, a multiplicative FPRAS exists. We will start by showing that in this case, the gap property holds: if the Shapley value is nonzero, then it must be at least the reciprocal of a polynomial.

% \benny{Why is the following a claim, whereas the others are lemmas?}
\begin{lemma}\label{claim:gap}
Let $q$ be a fixed CRPQ without redundancy. If $L(r_i)$ is finite for every atom $i$ of $q$, then $\pqshapleyval{q}(G, \tup u$, e) is either zero or at least $1/p(\abs{E})$.
\end{lemma}

\begin{proof}

If there is no subset $S$ of $E\endo$ such that adding $e$ to it along with $E\exo$ changes the value of query $q$ from false to true, then $\pqshapleyval{q}(G, \tup u, e) = 0$. Otherwise, let $S$ be a minimal such set, 
it holds that $\abs{S} \leq k_1 + \ldots + k_m = k$, where $k_i$ is the length of the longest word in $L(r_i)$; the language for the $i$-th atom in $q$, as at worst case, the paths match the longest word for each RPQ. And since each $L(r_i)$ is finite, each $k_i$ is a finite constant. Thus, $k$ also is a finite constant.

The probability to choose a permutation $\sigma$, such that $\sigma_e$ is exactly $S \setminus \{e\}$ is $$\frac{(\abs{S}-1)!(m_n-\abs{S})!}{m_n!} \allowbreak \geq \frac{(m_n-k)!}{m_n!}\,.$$
Hence, we have
\begin{align*}
\pqshapleyval{q}(G, \tup u, e) \geq  \frac{(m_n-k))!}{m_n!}  \frac{1}{(m_n-k+1)\cdot\ldots\cdot m_n}
= \frac{1}{p(\abs{E})}
\end{align*}
since $\abs{E}=m_n + m_e$.
\end{proof}

\begin{lemma}\label{lemma:approx_tract}
Let $q$ be a CRPQ without redundancy. If $L(r_i)$ is finite for every atom $i$ of $q$, then $\crpqshapley{q}$ has both an additive and a multiplicative FPRAS.
\end{lemma}
% \begin{theorem} 
% $\rpqshapley{r}$, where $L(r)$ is finite, has an FPRAS.
% \end{theorem}
\begin{proof}
Using the Chernoff-Hoeffding bound, we can get an additive FPRAS of the value $\pqshapleyval{q}(G, \tup u, e)$, by simply taking the ratio of successes over $O(\log(1/\delta)/\epsilon^2)$ trials of the following experiment:
% (Given an input instance $(G, s, t, e)$):
\begin{itemize}
    \item Select a random permutation $(e_1, ..., e_{m})$ over the set of endogenous edges $E\endo$.
    \item Suppose that $e=e_i$, and let $E_{i-1}=\{e_1, ..., e_{i-1}\}$. If
    $q[\tup u](G[E_{i-1} \cup E\exo \cup \{e\}])=1$ and 
    $q[\tup u](G[E_{i-1} \cup E\exo])=0$, then report ``success,'' otherwise, report ``failure.''
\end{itemize}

Now from Lemma~\ref{claim:gap} (that the gap property holds), we can easily get that an additive FPRAS also serves as a multiplicative one.
% we will prove that this also serves as a multiplicative FPRAS by showing that $\pqshapleyval{q}(G, (u_1, \ldots u_k), e)$ can not be very small.
\end{proof}

\begin{proof}[Proof of Theorem~\ref{thm:approx}]
Lemma~\ref{lemma:approx_hardness} shows the hardness side, as it implies that under conventional complexity assumptions, there is no polynomial-time multiplicative approximation when there is an atom with an infinite language (as it would allow to determine whether the Shapley value is nonzero).  Lemma~\ref{lemma:approx_tract} shows an FPRAS for the tractable case where all atom have finite languages.
\end{proof}

%% file: graphics/figure_reduction.tex
\usetikzlibrary{automata, positioning, arrows, calc}
\tikzset{
->, % makes the edges directed
>=stealth', % makes the arrow heads bold
node distance=11mm, % specifies the minimum distance between two nodes. Change if necessary.
every state/.style={thick, fill=gray!10, minimum size=3mm}, % sets the properties for each ’state’ node
initial text=$ $, % sets the text that appears on the start arrow
}
\tikzstyle{main2}=[circle,fill=black!50,minimum size=4pt,inner sep=0pt]

\begin{figure}[t]
\centering
\begin{subfigure}[b]{0.5\linewidth}
\centering
\begin{tikzpicture}
\node[state, initial] (A) {$A$};
\node[state, right of=A] (B) {$B$};
\node[state, below right of=B] (C) {$C$};
\node[state, right of=C] (E) {$E$};
% \node[state, right of=E] (F) {$F$};
\node[state, accepting, above right of=E] (D) {$D$};
% \node[state, accepting, right of=D] (G) {$G$};
\draw (A) edge[below, midway, sloped] node{a} (B)
(B) edge[below, midway, sloped] node{a} (C)
(C) edge[bend left=50, below, midway, sloped] node{b} (E)
(E) edge[bend left=50, below, midway, sloped] node{a} (C)
(E) edge[below, midway, sloped] node{c} (D)
% (E) edge[below, midway, sloped] node{b} (F)
% (F) edge[below, midway, sloped] node{c} (D)
% (D) edge[below] node{a} (G)
(B) edge[bend left, below, midway, sloped] node{c} (D)
% (A) edge[below, midway, sloped] node{a} (B)
% (C) edge[bend left=70, below, midway, sloped] node{b} (F)
% (F) edge[bend left=70, below, midway, sloped] node{a} (C)
% (E) edge[loop below, midway, sloped] node{a} (E)
;
% (q1) edge[above] node{1} (q2)
% (q2) edge[loop above] node{1} (q2)
% (q2) edge[bend left, above] node{0} (q3)
% (q3) edge[bend left, below] node{0, 1} (q2);
\end{tikzpicture}
\caption{The DFA for regular expression $a(a+b)^*c$.}
\vspace{1em}
\end{subfigure}\hfill
\begin{subfigure}[b]{0.5\linewidth}
\centering
\begin{tikzpicture}[node distance={12mm}, thick, main/.style = {draw, circle, minimum size=15pt, scale=0.8}] 
\node[main] (1) {$v_1$}; 
\node[main] (2) [below right of=1] {$v_2$}; 
\node[main] (3) [above right of=1] {$v_3$}; 
\node[main] (4) [above right of=2] {$v_4$}; 
% \node (13) at ($(1)!0.5!(3)$){4};

\draw[->] (1) -- (2); 
\draw[->] (1) -- (3); 
\draw[->] (2) -- (3); 
\draw[->] (2) -- (4); 
\draw[->] (3) -- (4); 

\end{tikzpicture} 
\caption{Input graph $G$.}
\vspace{1em}
\end{subfigure}

\begin{subfigure}[b]{\linewidth}
\centering
\begin{tikzpicture}
[node distance={15mm}, thick, main/.style = {draw, circle, minimum size=15pt,inner sep=0pt, scale=0.9}] 
\node[main] (A) {$v_A$}; 
\node[main] (B) [right of=A] {$v_B$}; 
\node[main] (1) [right of=B] {$v_1'$};
\node[main] (2) [below right of=1] {$v_2'$}; 
\node[main] (3) [above right of=1] {$v_3'$};
\node[main] (4) [above right of=2] {$v_4'$}; 
\node[main] (E) [right of=4] {$v_E$}; 
% \node[main] (F) [right of=E] {$v_F$}; 
\node[main] (D) [right of=E] {$v_D$}; 
% \node[main] (G) [right of=D] {$v_G$}; 

\node[main2] (13) at ($(1)!0.5!(3)$) {};
\node[main2] (12) at ($(1)!0.5!(2)$) {};
\node[main2] (23) at ($(2)!0.5!(3)$) {};
\node[main2] (24) at ($(2)!0.5!(4)$) {};
\node[main2] (34) at ($(3)!0.5!(4)$) {};

\draw (A) edge[above] node{a} (B)
(B) edge[above, midway, sloped] node{a} (1)
(4) edge[above, midway, sloped] node{b} (E)
(E) edge[above, midway, sloped] node{c} (D)
% (F) edge[above, midway, sloped] node{c} (D)
% (D) edge[below] node{a} (G)

(1) edge[above, midway, sloped] node{b} (13)
(13) edge[above, midway, sloped] node{a} (3)

(1) edge[below, midway, sloped] node{b} (12)
(12) edge[below, midway, sloped] node{a} (2)

(2) edge[above, midway, sloped] node{b} (23)
(23) edge[above, midway, sloped] node{a} (3)

(2) edge[below, midway, sloped] node{b} (24)
(24) edge[below, midway, sloped] node{a} (4)

(3) edge[above, midway, sloped] node{b} (34)
(34) edge[above, midway, sloped] node{a} (4)
;

% \draw[->] (1) -- (2); 
% \draw[->] (1) -- (3); 
% \draw[->] (2) -- (3); 
% \draw[->] (2) -- (4); 
% \draw[->] (3) -- (4); 

\end{tikzpicture} 
\caption{The graph $G'$ of the reduction for input instance $(G, v_1, v_4, e)$.}
\end{subfigure}
\caption{An example for the construction in the reduction of the proof of Lemma \ref{lemma:hardness_approx}}.
\label{fig:reduction_example}
\end{figure}
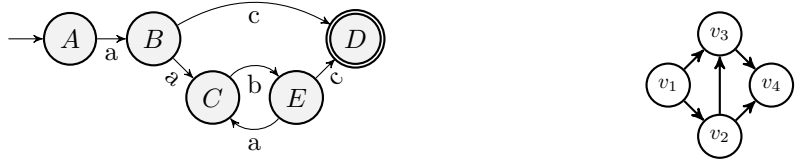
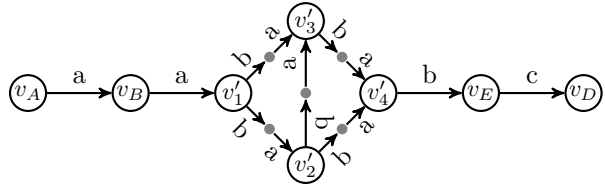

%% file: vertices.tex
\section{Shapley Value of Vertices}\label{sec:vertices}
\label{chap:vertices}

In this section, we discuss the differences between the computation complexity of the Shapley value for edges and the Shapley values for vertices in the graph.
Similarly to the case for edges, given a conjunctive regular path query $q$, our goal is to quantify the contribution of vertices in the input graph $G$ to an answer of the path query. The graph consists of two types of vertices---\e{endogenous vertices}, and \e{exogenous vertices}.

Notationally, for a graph $G = (V, E)$ we denote by $V\endo$ and $V\exo$ the sets of endogenous and exogenous vertices, respectively,
and we assume that $V$ is the disjoint union of $V\endo$ and $V\exo$. We denote by $\pqshapleyval{q}(G,\tup u, v)$ the Shapley value of a vertex $v \in V\endo$.
\[ \pqshapleyval{q}(G,\tup u, v) \eqdef \shapley(V\endo, \vpqv, v) \]
% \benny{Bad notation... it seems like the same $v$ in the base and the superscript.}
Where $\vpqv$ is defined as follows:
\[
\vpqv(B) \eqdef q[\tup u] (G[B \cup V\exo])-q[\tup  u](G[V\exo])
\]

We denote by $\crpqshapleyv{q}$ and $\rpqshapleyv{r}$, The corresponding computational problems to those defined earlier for the Shapley values of edges. We now state the results we have with some notes on the changes that should be made in the proofs.

\subsection{Complexity of Exact Computation}
The hardness part is almost the same as the case of edges. We begin with hardness for the special case where the regular language (or any language) consists of a single four-letter word instead of three. For that, we use the same result by Livshits et al.~\cite{DBLP:journals/sigmod/LivshitsBKS21} on the computation of Shapley values for facts in relational databases. From this we continue the same series of reductions as done for edges to get the hardness for a general CRPQ.
The tractable part is also tractable when looking at vertices. So we have the following:

\begin{theorem}\label{thm:exact_hard_vertex}
The following hold for a CRPQ $q$. 
\begin{enumerate}
    \item 
    If $q$ has a non-redundant atom $i$ with a language that  contains a word of length four or more, then $\crpqshapley{q}$ is $\fpsharpp$-complete.
% If \benny{change phrasing} there exists an atom $i$ of $q$ such that it is not redundant and $L(r_i)$ contains a word of length at least four, then $\crpqshapleyv{q}$ is $\fpsharpp$-complete.
    \item If $q$ has only one atom with regular expression $r$. If every word in $L(r)$ is of length at most two, then $\crpqshapleyv{q}$ is solvable in polynomial time.
\end{enumerate}
\end{theorem}

Note that in the case of vertices, we leave a gap in the classification of RPQs.  Theorem~\ref{thm:exact_hard_vertex} states that if there exists a word of length four or more, then the problem is hard, and if all words are of length at most two, then the problem is solvable in polynomial time. The case where there are words of length three but not longer remains an open problem (as opposed to the case of edges where we had a full dichotomy on RPQs due to Corollary~\ref{cor:exact-rpq-dichotomy}).

% \begin{theorem}[Tractability]\label{thm:exact_tract_vertex}
% Let $q$ be a CRPQ with one atom with regular expression $r$. If every word in $L(r)$ is of length at most two, then $\crpqshapleyv{q}$ is solvable in polynomial time.
% \end{theorem}
% \benny{Not good. Combine these two theorems into one nice theorem.}

%\bennyc{Say something about the remaining gap!!!}

\subsection{Complexity of Approximation}
For approximation, we get the exact same dichotomy on CRPQs. We know from before that also the decision problem that decides whether a vertex $v$ lies on a simple path from $s$ to $t$ in graph $G$, is $\np$-complete. From that we get that the problem of determining whether $\pqshapleyval{\Sigmastar}(G, s, t, v) > 0$ is also $\np$-complete. From that we continue with a series of reductions that is almost identical to what we have for the case of edges.

\begin{theorem}\label{thm:approx_vertex}
  Let $q$ be a CRPQ without redundancy. If $L(r_i)$ is finite for every atom  $i$ of $q$ , then $\crpqshapleyv{q}$ has a multiplicative FPRAS.
  Otherwise, $\crpqshapleyv{q}$ has no multiplicative approximation (of any ratio) or else $\mbox{NP}\subseteq\mbox{BPP}$.
\end{theorem}

% \begin{theorem}\label{thm:approx_vertex}
% Let $q$ be a CRPQ. If there exists an atom $i$ of $q$ such that it is not redundant and $L(r_i)$ is infinite, then $\crpqshapleyv{q}$ has no FPRAS. Otherwise, it has an FPRAS.
% \end{theorem}
% \benny{Change the phrasing of the theorem to be the same as in the edge case.}

% \benny{Again, if we have NP-completeness then say NP-complete, not NP-hard.}

In conclusion, we establish that the complexity for both exact computation and approximation of the Shapley value of vertices is very similar to the case of edges.
It is generally hard to compute exact values; it is sufficient for the CRPQ to have an atom that is non-redundant and contains a word of length four or more for the computation to be hard, while for RPQs we identify that the tractable family of queries for edges is also tractable for vertices. For approximation, we show that we have an identical dichotomy for when queries admit a multiplicative FPRAS.

%% file: conclusion.tex
\section{Concluding Remarks}
\label{sec:conclusion}

This work continues the research line of responsibility and contribution in databases. We presented the graph-database perspective where the queries are (conjunctive) regular path queries, and the responsibility measure is the Shapley value.  We investigated the data complexity of the Shapley value of edges in the graph. For the exact computation, we showed that it is generally hard, while we also show a specific family of CRPQs where the computation can be done in polynomial time. While this is not a full dichotomy on CRPQs, the tractable case we showed basically defines a dichotomy on the class of RPQs. 
It remains an open problem whether the condition we have for hardness defines a full dichotomy on CRPQs.
% We conjecture that the condition we have for hardness should define a dichotomy on CRPQs, but that is yet to be proved.\benny{Risky... why do we believe that?}  
We have also studied approximation of the Shapley values in the form of an FPRAS. An additive FPRAS is easy to achieve using Monte-Carlo sampling, while a multiplicative approximation is harder. We showed a family of CRPQs where the gap property holds, and hence, an additive FPRAS can be transformed into a multiplicative one. These are the CRPQs where every 
% (non-redundant) 
atom has a finite language. For the other CRPQs, we showed that it is hard to obtain any multiplicative approximation. Thus, we achieved a dichotomy on CRPQs for the case of approximation (assuming no redundant atoms).

% dichotomy on RPQs

Several problems remain open. We still do not have a full dichotomy for exact computation of Shapley values.  In addition, the proof of the hardness of approximation in Section~\ref{sec:approx_hard} is not valid when the graph in hand is acyclic; this raises the question of whether there are better opportunities of efficient approximations when the problem is restricted to acyclic graphs. Another direction is investigating richer path languages, for example, allowing existentially quantified variables in the query, or negated atoms.
% \benny{existentially quantified variables...}